\documentclass[usletter,12pt,leqno]{article}

\usepackage{graphicx}
\usepackage{amsmath,amsthm,amssymb,mathrsfs,mathtools}
\usepackage[OT1]{fontenc} 
\usepackage{etoc}
\usepackage{float}
\usepackage{fancyhdr}

\def\subsectiontitle{}
\def\subsubsectiontitle{}
\usepackage{tikz}
\usetikzlibrary{automata, positioning}

\usepackage[noamsmath]{kpfonts}
\usepackage{zi4}

\usepackage{enumitem}
\usepackage[margin=1in]{geometry}
\usepackage[colorlinks,breaklinks=true,pagebackref]{hyperref}
\usepackage{breakcites}

\usepackage{xcolor}
\AtBeginDocument{%
	\hypersetup{%
		linkcolor={red!50!black},%
		citecolor={blue!50!black},%
		urlcolor={blue!80!black}%
	}%
}%

\usepackage[nameinlink,noabbrev,sort,capitalise]{cleveref}
\makeatletter
\def\ps@pprintTitle{%
	\let\@oddhead\@empty
	\let\@evenhead\@empty
	\def\@oddfoot{\emph{Very preliminary version}\hfill\emph{This draft: \today}}%
	\let\@evenfoot\@oddfoot}
\makeatother
\usepackage{etoolbox}
\patchcmd{\pprintMaketitle}
{\fi\hrule}
{\fi\ifvoid\extrainfobox\else\unvbox\extrainfobox\par\vskip10pt\fi\hrule}
{}{}

\newsavebox\extrainfobox

\usepackage{comment}
\usepackage{multirow}
\usepackage{multicol}
\usepackage{subcaption}

\crefname{problem}{Problem}{Problems}

\usepackage{indentfirst}

\allowdisplaybreaks

\let\oldfootnote\footnote
\renewcommand\footnote[1]{\oldfootnote{\hspace{.4mm}#1}}

\makeatletter
\renewenvironment{proof}[1][\proofname] {\par\pushQED{\qed}\normalfont\topsep6\p@\@plus6\p@\relax\trivlist\item[\hskip\labelsep\bfseries#1\@addpunct{.}]\ignorespaces}{\popQED\endtrivlist\@endpefalse}
\makeatother

\let\oldFootnote\footnote
\newcommand\nextToken\relax

\renewcommand\footnote[1]{%
	\oldFootnote{#1}\futurelet\nextToken\isFootnote}

\newcommand\isFootnote{%
	\ifx\footnote\nextToken\textsuperscript{,}\fi}

\usepackage{blkarray}
\usepackage{graphicx,pgfpages,epsfig,multirow,lscape}
\usepackage{hyperref}
\usepackage{epic}
\usepackage{xcolor}
\usepackage{setspace}
\usepackage{tikz}
\usetikzlibrary{arrows,decorations,decorations.pathreplacing,calc,matrix}
\usepackage{natbib}
\DeclareFontFamily{U}{mathb}{\hyphenchar\font45}
\DeclareFontShape{U}{mathb}{m}{n}{
	<-6> mathb5 <6-7> mathb6 <7-8> mathb7
	<8-9> mathb8 <9-10> mathb9
	<10-12> mathb10 <12-> mathb12
}{}
\DeclareSymbolFont{mathb}{U}{mathb}{m}{n}
\DeclareMathSymbol{\llcurly}{\mathrel}{mathb}{"CE}
\DeclareMathSymbol{\ggcurly}{\mathrel}{mathb}{"CF}

\def\gg{\ggcurly}

\def\w{\omega}
\def\a{\mathcal{A}}

\def\o{ o^*}
\def \r{\mathcal{R}}

\hypersetup{colorlinks = true}
\hypersetup{allcolors=blue}

\newtheorem{definition}{Definition}

\newtheorem{theorem}{Theorem}
\newtheorem*{theorem*}{Theorem}
\newtheorem{proposition}{Proposition}
\newtheorem{lemma}{Lemma}
\newtheorem{example}{Example}
\newtheorem{corollary}{Corollary}
\newtheorem{remark}{Remark}
\newtheorem{assumption}{Assumption}

\newcommand{\norm}[1]{\| #1 \|}

\begin{document}

\title{Strategy-proof allocation with outside option\thanks{I am grateful to seminar and conference participants in Waseda University, Shanghai University of Finance and Economics, 2018 Nanjing International Conference on Game Theory, and 2019 Conference on Economic Design in Budapest. This paper is supported by National Natural Science Foundation of China (Grant 71903093 and Grant 72033004). An earlier version was circulated under the title ``Truncation-invariant allocation.''
}}

\author{Jun Zhang\\
	\small 
	Institute for Social and Economic Research, Nanjing Audit University\\\small Email: zhangjun404@gmail.com}


\maketitle

\begin{abstract}
	Strategy-proof mechanisms are widely used in market design. In an abstract allocation framework where outside options are available to agents, we obtain two results for strategy-proof mechanisms. They provide a unified foundation for several existing results in distinct models and imply new results in some models. The first result proves that, for individually rational and strategy-proof mechanisms, pinning down every agent's probability of choosing his outside option is equivalent to pinning down a mechanism. The second result provides a sufficient condition for two strategy-proof mechanisms to be equivalent when the number of possible allocations is finite. 
\end{abstract}



\thispagestyle{empty}
\setcounter{page}{0}
\newpage

\section{Introduction}

In mechanism design, strategy-proof mechanisms are desirable because they incentivize agents to reveal private types.  This paper considers allocation problems where agents' private types are their preferences and they can choose outside options if they wish. Such problems are prevalent in market design and other economic environments. In marriage markets  men and women have the outside option of being alone \citep{gale1962college}. In school choice  children have the outside option of attending private schools \citep{abdulkadiroglu2003school}. In auction markets  bidders have the outside option of not bidding \citep{vickrey1961counterspeculation}. This paper proves two simple but useful results for strategy-proof mechanisms in such problems. Our results treat deterministic mechanisms and random mechanisms in a unified way. They provide a unified foundation for several existing results that are derived by ad hoc approaches in distinct market design models. Uncovering this foundation not only deepens our understanding of these models and results, but also provides new tools and terminologies for future studies.

Specifically, we provide a unified foundation for the following results.
\begin{itemize}
	\item In a generalized object assignment model, \cite{sonmez1999strategy} proves that there exists an individually rational, Pareto efficient and strategy-proof mechanism only if the core is essentially single-valued and the mechanism is core-selecting when the core is nonempty. This finding unifies the results in several distinct models, including housing market, marriage and roommate problem, coalition formation problem, and social network. \cite{ehlers2018strategy} extends this finding to a superset of the core called the individually-rational-core. 
	
	\item In the object assignment model, \cite{erdil2014strategy} proves that, among individually rational and strategy-proof random mechanisms, one mechanism stochastically dominates another mechanism only if the former assigns more probability shares of objects in total to agents than the latter. It implies that non-wasteful and strategy-proof mechanisms cannot be strategy-proof dominated.
	
	\item In the many-to-one matching with contracts model, under very weak assumptions \cite{hirata2017stable} prove that there is at most one stable and strategy-proof mechanism, and the doctor-optimal stable mechanism, if it exists, is the only candidate. They also propose a notion of non-wastefulness (different than the usual one in the literature) and prove that any non-wasteful and strategy-proof mechanism is not Pareto dominated by any other individually rational and strategy-proof mechanism.
	
	\item In an abstract allocation model, \cite{alva2017strategy} prove that, in the class of individually rational and strategy-proof mechanisms, one mechanism Pareto dominates another if and only if the set of participants in the allocations found by the former mechanism is a superset of participants in the allocations found by the latter mechanism. This result has useful implications in models including object assignment, school choice, excludable public goods, and transferable utility.
	
	\item In a teacher reassignment model, \cite{combe2020design} propose a class of mechanisms of desirable properties, and prove that a specific mechanism they recommend is the only strategy-proof member of the class when the matching is one-to-one.
\end{itemize} 
Among the above results, Erdil's is about random allocation mechanisms in the object assignment model, while the others are about deterministic allocation mechanisms in other models. These results seem to be unrelated to each other. But we show that there is a unified logic behind all of them. 

We present the logic in an abstract allocation model that has been used by \cite{alva2017strategy}. The model consists of finite agents and a set of allocations. It is abstract because it does not explain what allocations mean. It only requires that each allocation have a well-defined set of participants (non-participants consume outside options), and agents have ordinal preferences over allocations. After eliciting agents' preferences, a deterministic mechanism returns an allocation, while a random mechanism returns a randomization over allocations. To treat the two types of mechanisms in a unified way, we present our results for random mechanisms and regard deterministic mechanisms as special cases. With agents' ordinal preferences, we use (first-order) stochastic dominance to make welfare comparison between random allocations/mechanisms. A mechanism is \textit{strategy-proof} if, for every agent, the random allocation under truth-telling stochastically dominates the random allocation under manipulation. A mechanism is \textit{individually rational} (IR) if it never selects an allocation that makes some agents worse off than consuming outside options with a positive probability.

Our first result, Theorem \ref{thm1}, proves that if agents can  freely vary the ranking of outside options in their preferences and are never indifferent between outside options and the other allocations, then among IR and strategy-proof mechanisms, a mechanism stochastically dominates another mechanism if and only if every agent has a weakly higher probability of being a participant in the former mechanism than in the latter mechanism, and the dominance is strict if and only if the participation probability is strictly higher for some agent in the former mechanism. This result looks surprising at first glance, because usually comparing welfare requires more information than comparing participation probabilities. Our insight is that, for an IR and strategy-proof mechanism, every agent's welfare in every preference profile can be recovered from his participation probability in all preference profiles after he varies the ranking of his outside option in preferences.

Theorem \ref{thm1} unifies the results of \cite{erdil2014strategy}  and \cite{alva2017strategy}(AM for short). When we restrict Theorem \ref{thm1} to deterministic mechanisms, it coincides with AM's theorem. While if we apply our framework to the object assignment model, an agent's participation probability in a random allocation is the total probability shares of objects he receives. Then Theorem \ref{thm1} means that, for IR and strategy-proof random mechanisms in the object assignment model, one mechanism stochastically dominates another mechanism if and only if the former mechanism assigns weakly more probability shares of objects to every agent, and the dominance is strict if and only if the former mechanism assigns strictly more probability shares of objects to some agent. This ``if and only if'' result is more precise than Erdil's result that strategy-proof dominance requires assigning more probability shares of objects in total to agents. We also derive properties weaker than non-wastefulness that are sufficient for an IR and strategy-proof mechanism to be undominated by other strategy-proof mechanisms.

Our second result, Lemma \ref{thm2}, presents a sufficient condition for two strategy-proof mechanisms to be equivalent when the set of allocations in our model is finite. With a precise statement left to the paper, it roughly proves that, when agents can freely vary the ranking of outside options in their preferences, for any two strategy-proof mechanisms, if whenever they assign equal probabilities to the allocations in the ``upper part'' of agents' preferences, their outcomes are equivalent, then the two mechanisms must be equivalent. Lemma \ref{thm2} will be useful in applications  for proving uniqueness of strategy-proof mechanisms among a class of mechanisms of certain desirable properties. As we show, it directly implies the uniqueness results proved by \cite{hirata2017stable} and  \cite{combe2020design} in their respective models. Although Lemma \ref{thm2} does not directly imply the results of \cite{sonmez1999strategy} and \cite{ehlers2018strategy}, it is the key to the proof of their results.

Our two results have a unified foundation. When agents vary the ranking of outside options in their preferences, a type of strategies attracting our attention is called \textit{contraction}: an agent moves his outside option upwards in preferences to a position that is immediately below an allocation (called \textit{contraction point}) and does not change the upper contour set of the allocation. That is, allocations that are weakly above than the contraction point remain weakly above than the contraction point in the contracted preferences. A special case of contraction is known as \textit{truncation} \citep{roth1999truncation}: an agent improves the ranking of his outside option without changing preferences over the other allocations.
In a strategy-proof mechanism, if an agent contracts preferences, because the upper contour set of the contraction point is invariant, the total probability mass of the upper contour set in the allocations found by the mechanism must also be \textit{invariant}. We call this property \textit{contraction-invariance}; when we only consider truncation, we call it \textit{truncation-invariance}. Contraction-invariance connects the allocations in difference preference profiles, which is the key to our results.  If agents can freely truncate preferences, our results hold for non-strategy-proof mechanisms that are truncation-invariant. In the matching with contracts model of \cite{hirata2017stable}, we prove that if the doctor-optimal stable mechanism exists, it must be truncation-invariant, although it may not be strategy-proof. This is enough to invoke Lemma \ref{thm2} to prove that the doctor-optimal stable mechanism is the only candidate for a stable and strategy-proof mechanism, if it exists.

Besides unifying existing results, we also provide a new application of Lemma \ref{thm2}. In a market such as kidney exchange and refugee resettlement, being unmatched is a very undesirable outcome for agents, and therefore a market designer often wants to increase the number of acceptable assignments in the market. When a mechanism is in use, to replace it with a new one, for political reasons we often need to respect agents' welfare in the existing mechanism. To model this, we propose a \textit{size improvement} relation between an existing mechanism and a new one in the object assignment model.  It requires that the new mechanism assign strictly more acceptable objects in total to agents in at least one preference profile, and when it does not do so in some preference profile, every agent who obtains an acceptable object in the existing mechanism still obtain an acceptable object in the new mechanism. Applying Lemma \ref{thm2}, we prove that size improvement is impossible if the two mechanisms are truncation-invariant and the existing mechanism is not too inefficient. This result is applicable to well-studied mechanisms in the literature, including strategy-proof mechanisms such as deferred acceptance, top trading cycle (and its generalizations), and serial dictatorship, as well as non-strategy-proof but truncation-invariant mechanisms such as the boston mechanism \citep{abdulkadiroglu2003school} and probabilistic serial \citep{bogomolnaia2001new}.

After discussing related literature in the rest of this section, we present our model in Section \ref{section:model}, our results and their corollaries in Section \ref{section:theorems}, and applications in Section \ref{section:application}. For ease of exposition, we assume a finite set of allocations in our model. This assumption is relied upon by Lemma \ref{thm2} and satisfied by our applications. In Remark \ref{rmk:thm1:infinite}, we explain why Theorem \ref{thm1} holds without change when the set of allocations is infinite after relevant concepts are properly rephrased. We conclude the paper in Section \ref{section:conclusion}. Appendix contains omitted proofs and results.

\paragraph{Related literature.} Outside options enrich agents' manipulation strategies and strengthen the restriction of strategy-proofness. Our results reflect this strengthened restriction. Several results in the literature have similarly relied on the existence of outside options. In school choice with coarse priorities, \cite{abdulkadiroglu2009strategy} prove that when students have outside options, no strategy-proof deterministic mechanism Pareto improves the deferred acceptance algorithm (after ties are broken). As \cite{alva2017strategy} have explained, this result is implied by Theorem \ref{thm1}. \cite{kesten2016outside} prove that even if exogenous outside options are absent, as long as students can arbitrarily rank schools in preferences, some schools will endogenously  play the role of outside options.\footnote{When priorities are strict and students may not have outside options, \cite{kesten2010school} has proved that no strategy-proof and Pareto efficient deterministic mechanism Pareto improves deferred acceptance.}

When agents can freely truncate preferences, our results hold for truncation-invariant mechanisms. Various invariance notions have been proposed in the literature. In the object assignment model, \cite{hashimoto2014two} define \textit{weak truncation robustness} and use it to characterize probabilistic serial. In Remark \ref{rmk:object:truncation-invariance} we explain that it is equivalent to our truncation-invariance in the object assignment model. \cite{ehlers2014strategy,ehlers2016object} define a different invariance notion for deterministic mechanisms. It requires that if an agent truncates preferences at a point that is worse than his assignment, then the allocation for the whole market remain unchanged. This requirement is much stronger than truncation-invariance. In the matching with contracts model, \cite{hatfield2019stability} define \textit{truncation-consistency} for deterministic mechanisms. It requires that if an agent truncates preferences at a point that is worse than his assignment, then his assignment remain unchanged. Truncation-consistency essentially coincides with truncation-invariance for deterministic mechanisms.

In the object assignment model we propose a size improvement relation between a mechanism that is in use and a different mechanism that is supposed to replace the former. \cite{Afacandurpossible} (and also \cite{Afacanmaximization}) propose a different criterion to compare the number of acceptable assignments in two deterministic mechanisms. In their criterion the positions of the two mechanisms are symmetric. They prove two impossibility results in the same flavor as ours. Compared with their results, ours uses weaker axioms, allows for weak preferences, and is applicable to deterministic mechanisms as well as random mechanisms. Our results are independent of each other.

\section{Model}\label{section:model}

Our model consists of a finite set of agents $ I $  and a finite set of allocations $ \a $. We denote an agent by $ i $ or $ j $, and an allocation by $ a $ or $ b $. Each allocation $ a $ is associated with a set of \textit{participants} $ I(a)\subset I $. For each agent $ i $, $ \a_i\coloneqq \{a\in \a: i\in I(a)\} $ is the set of allocations where $ i $ is a participant, and $ i $ consumes his outside option $ \o_i $ in every $a \in \a\backslash \a_i $. 

Each agent $ i $ has a complete and transitive preference relation $ \succsim_i $ on $ \a_i \cup \{\o_i\} $. So we assume that $ i $ is indifferent between any two allocations where he consumes his outside option. Note that $ i $'s preferences over the other allocations can be strict or weak; for example, $ i $ can be indifferent between some $ a\in \a_i $ and some $ b\in \a\backslash \a_i  $. Let $ \succ_i $ and $ \sim_i $ respectively denote the asymmetric component and symmetric component of $ \succsim_i $. Let $ \mathcal{R}_i $ denote the domain of $ i $'s preferences. Let $ \succsim_I\coloneqq(\succsim_i)_{i\in I} $ denote a preference profile, and $ \mathcal{R}\coloneqq\times_{i\in I} \mathcal{R}_i $ denote the domain of preference profiles. 

An allocation $ a$ is \textit{acceptable} to $ i $ if $ a\succsim_i \o_i $, and is \textit{strictly acceptable} to $ i $ if $ a\succ_i \o_i $. An allocation $ a $ is \textit{individually rational (IR)} if $ a\succsim_i \o_i $ for all $ i$. An allocation $ b$ \textit{weakly Pareto dominates} another allocation $ a $ if $ b\succsim_i a $ for all $ i $; if there further exists $ j$ such that $ b\succ_j a $, then $ b $ \textit{strictly Pareto dominates} $ a $; if $ a\sim_i b $ for all $ i $, then $ a $ and $ b $ are \textit{welfare-equivalent}. An allocation is \textit{Pareto efficient} if it is not strictly Pareto dominated by any other allocation.  

\cite{alva2017strategy} have demonstrated that this abstract model can describe many economic environments. We will use the object assignment model as a running example to illustrate various concepts we define for the abstract model. Remark \ref{rmk:object:defn} briefly defines the object assignment model. 

\begin{remark}\label{rmk:object:defn}
	In the object assignment model, a finite set of indivisible objects $ O $ are assigned to a finite set of agents $ I $. Each object $ o $ has $ q_o \in \mathbb{N}$ copies. Each agent $ i $ demands a copy of an object, and has a preference relation $ \succsim_i $ on $ \tilde{O}\coloneqq O\cup\{\emptyset\}  $ where $ \emptyset $ denotes obtaining nothing. An allocation is a function $ f:I\rightarrow \tilde{O} $ such that, for all $ o\in O $, $ |f^{-1}(o)|\le q_o $. 	
	To fit it into our abstract model, let $ \emptyset $ be every agent's outside option and let the participants of an allocation be those who are assigned objects. An agent's preferences over allocations are induced by his preferences over his assignments in the allocations.
\end{remark}

For every $ a\in \a $ and $ \succsim_i\in \r_i $, we define following notations:
\begin{itemize}
	\item $ \a(\succsim_i, a)\coloneqq\{b\in \a:b\succsim_i a\} $ is the \textit{upper contour set} of  $ a $ in $ \succsim_i $;
	
	\item $ \a(\succ_i, a)\coloneqq\{b\in \a:b\succ_i a\} $ is the \textit{strict upper contour set} of $ a $ in $ \succsim_i $;
	
	\item $ \a(\sim_i, a)\coloneqq \{b\in \a: b\sim_i a\} $ is the \textit{indifference class} of $ a $ in $ \succsim_i $;
	
	\item $ \a_i(\succsim_i, a) $, $ \a_i(\succ_i, a) $, and $ \a_i(\sim_i, a) $ are similarly defined.
\end{itemize}
It is worth mentioning that $ \a(\succ_i, \o_i) $ is the set of strictly acceptable allocations for $ i $, $\a(\succ_i, \o_i)=\a_i(\succ_i, \o_i)  $, and $ \a_i(\sim_i,\o_i) $ is the set of allocations that $ i $ participates in and regards as indifferent with $ \o_i $. For every $ \succsim_i\in \r_i $, we also define
\[
\a_i(\gg_i,\o_i)\coloneqq\{a\in \a_i:\exists b\in\a_i \text{ such that }a\succ_i b\succsim_i \o_i  \}
\] and call it  the \textit{upper part} of $ \a_i(\succsim_i,\o_i) $.
If $ \a_i(\sim_i,\o_i) \neq\emptyset$, then $ \a_i(\gg_i,\o_i)=\a_i(\succ_i,\o_i) $, while if  $ \a_i(\sim_i,\o_i) =\emptyset$, then $ \a_i(\gg_i,\o_i) $ is the set of strictly acceptable allocations that are not immediately above $ \o_i $ in $ \succsim_i $. Understanding $ \a_i(\gg_i,\o_i) $ is crucial for understanding our analysis.

\subsection{Random allocation}\label{section:random:allocation}

We regard every element of $ \a $ as a \textit{deterministic} allocation. A \textit{random allocation} is a probability distribution on $ \a $. In a random allocation $ p $, $ p_a $ is the probability of every $ a\in \a $, and $ p[A]\coloneqq \sum_{a\in A}p_a $ is the probability of every $ A\subset \a $. Let $ \Delta(\a) $ denote the set of random allocations. Because every deterministic allocation can be regarded as a degenerate random allocation, we make the convention that $ \a\subset \Delta(\a)$ and simply call the elements of $ \Delta(\a) $ \textit{allocations} in the rest of the paper. To avoid confusion we will call elements of  $ \a $ \textit{deterministic allocations}. 

An allocation $ p $ is \textit{individually rational (IR)} if, for every $ a $ with $ p_a>0 $, $ a\succsim_i \o_i $ for every $ i $. In other words, $ p $ is IR if and only if $p[\a(\succsim_i, \o_i)]=1  $ for every $ i $.

Given agents' ordinal preferences over deterministic allocations, we use (first-order) stochastic dominance to compare random allocations.
\begin{definition}[Stochastic dominance]
	For any two allocations $ p $ and $ p' $,
	\begin{enumerate}
		\item $ p' $ \textbf{weakly stochastically dominates} $ p $ \textbf{for agent $ i $}, denoted by $ p'\succsim^{sd}_i p$, if
		\[
	\sum_{b\succsim_i a}p'_a \ge \sum_{b\succsim_i a}p_a \text{ for all }a\in \a.
		\]
		If the inequality is strict for some $ a $, then $ p' $ \textbf{strictly stochastically dominates} $ p $ \textbf{for $ i $}, denoted by $ p'\succ^{sd}_i p$.	
	   If no inequalities are strict, then $ p' $ and $ p $ are \textbf{welfare-equivalent for $ i $}, denoted by $ p'\sim^{sd}_i p$.
		
		\item $ p' $ \textbf{weakly stochastically dominates} $ p $ if $ p'\succsim^{sd}_i p $ for all $ i $. 		
		If there exists $ j$ such that $ p'\succ^{sd}_j p $, then $ p' $ \textbf{strictly stochastically dominates} $ p $.
		If $ p'\sim^{sd}_i p $ for all $ i $, then $ p' $ and $ p $ are \textbf{welfare-equivalent}.
	\end{enumerate}
\end{definition}

For every allocation $ p $ and agent $ i $, $ p[\a_i]$ is $ i $'s probability of being a participant in $ p $. This probability will play an important role in Theorem \ref{thm1}.

\begin{definition}[Participation size]
	For every allocation $ p $ and agent $ i $, $ p[\a_i]$ is called \textbf{$ i $'s participation size} in $ p $.
\end{definition}

In a deterministic allocation, an agent's participation size is one if he is a participant, and otherwise it is zero.

\begin{remark}\label{rmk:object:randomallocation}
	In the object assignment model, thanking the Birkhoff-von Neumann theorem,  a random allocation can be represented by a nonnegative matrix $ p=(p_{i,o})_{i\in I, o\in \tilde{O}} $ such that $ \sum_{i\in I}p_{i,o}\le q_o $ for all $ o\in O $ and $ \sum_{o\in \tilde{O}}p_{i,o}= 1 $ for all $ i\in I $.
	Every $ p_{i,o} $ denotes the probability that $ i $ obtains a copy of $ o $, and $ p_i=(p_{i,o})_{o\in \tilde{O}} $ is the lottery assigned to $ i $. Every $ i $'s participation size in a random allocation $ p $ is the total probability share of objects $ i $ receives (i.e., $ \sum_{o\in O}p_{i,o} $).
\end{remark}

We compare agents' participation sizes in different allocations.

\begin{definition}[Size dominance]
	For any two allocations $ p $ and $ p' $, $ p' $ \textbf{weakly size dominates} $ p $ if $ p'[\a_i]\ge  p[\a_i] $ for all $ i $. If the inequality is strict for some $ i $, then $ p' $ \textbf{strictly size dominates} $ p $. If $ p'[\a_i]=  p[\a_i] $ for all $ i$, then $ p' $ and $ p $ are \textbf{size-equivalent}.
\end{definition}

For two deterministic allocations $ a $ and $ b $, $ a $ weakly size dominates $ b $ if $ I(a)\supset I(b) $, and they are size-equivalent if $ I(a)=I(b) $.

Size dominance is independent of stochastic dominance, because the former does not refer to agents' preferences. But for IR allocations, an agent's participation size contains information about his welfare. Specifically, for every IR allocation $ p $ and every $ i $, 
\[
p[\a_i]=p[\a_i(\succ_i,\o_i)]+p[\a_i(\sim_i,\o_i)].
\]
Note that $ p[\a_i(\succ_i,\o_i)] $ is the probability that $ i $ is better off than consuming his outside option. So $ p[\a_i] $ is an upper bound on $ p[\a_i(\succ_i,\o_i)] $. If $ \a_i(\sim_i,\o_i)=\emptyset $, then $ p[\a_i] $ is equal to $ p[\a_i(\succ_i,\o_i)] $. If $ \a_i(\sim_i,\o_i)=\emptyset $ for all $ i $, then weak stochastic dominance implies weak size dominance (yet strict stochastic dominance still does not imply strict size dominance). In environments where agents regard allocations they participate in as fundamentally different from outside options, $ \a_i(\sim_i,\o_i)=\emptyset $ holds. For example, kidney patients may regard compatible donors as indifferent, but they do not regard compatible donors as indifferent with no transplant. Of course, $ \a_i(\sim_i,\o_i)=\emptyset $ holds when agents' preferences are strict. Below we formalize it as an assumption on agents' preference domains. 

\begin{assumption}[No indifference with outside option (NI)]
	For every $ i $ and $ \succsim_i\in \mathcal{R}_i $, $ \a_i(\sim_i,\o_i)=\emptyset $.
\end{assumption}

Our results rely on agents' freedom of varying the ranking of outside options in preferences. This is formalized as a richness assumption on their preference domains. Recall that for every $ \succsim_i $, $ \a_i(\gg_i,\o_i) $ is the ``upper part'' of $ \a_i(\succsim_i,\o_i) $.  

\begin{assumption}[Richness]
	For every $ i $, $ \succsim_i\in \mathcal{R}_i $, and $ a\in \mathcal{A}_i(\gg_i,\o_i)$, there exists $ \succsim'_i\in \mathcal{R}_i $ such that, for all $ b\in \a_i $, $ b\succsim_i a \implies b\succsim'_i a$, and $ a\succ_i b \implies a\succ'_i \o_i \succ'_i b $.
\end{assumption}

Put differently, Richness requires that for every $ \succsim_i\in \mathcal{R}_i $ and $ a\in \mathcal{A}_i(\gg_i,\o_i)$, there exist $ \succsim'_i\in \r_i $ such that 
\[
 \a_i(\succsim_i ,a)=\a_i(\succsim'_i,a)=\a_i(\succ'_i,\o_i) \text{ and }\a_i(\sim'_i,\o_i)=\emptyset.
\] In words, $ \o_i $ is ranked immediately below $ a $ in $ \succsim'_i $, the upper contour set of $ a $ is invariant from $ \succsim_i $ to $ \succsim'_i $, and no $ a\in \a_i $ is indifferent with $ \o_i $ in $ \succsim'_i $ (so NI is assumed for $ \succsim'_i $). Because there are fewer acceptable allocations in $ \succsim'_i $ than in $ \succsim_i $, we call $ \succsim'_i $ a \textit{contraction of $ \succsim_i $ at $ a $}.

\begin{definition}(Contraction)
	For every $ i $, $ \succsim_i\in \mathcal{R}_i $ and $ a\in \mathcal{A}_i(\gg_i,\o_i)$, every $ \succsim'_i $ in Assumption 2 is called a \textbf{contraction of $ \succsim_i $ at $ a $}. If a contraction $ \succsim'_i $
	ranks the allocations in $ \a_i $ in the same way as $ \succsim_i $ does, we call it the \textbf{truncation of $ \succsim_i $ at $ a $}, and denote it by $ \succsim^a_i $.
\end{definition}

  We say $ \r_i $ accommodates truncation if it contains the truncation of every $ \succsim_i\in \mathcal{R}_i $ at every $ a\in \mathcal{A}_i(\gg_i,\o_i)$.
  
 \begin{remark}\label{rmk:object:prefdomain}
 	In the object assignment  model, Richness and NI are satisfied if every $ \r_i $ consists of all strict preference relations, or if every $ \r_i $ consists of all preference relations in which there is no indifference between real objects and outside options. We will use these two preference domains in Section \ref{section:application}.
 \end{remark}

Finally, we define a notion called \textit{upper-equivalence}, which coarsens welfare-equivalence. Recall that two allocations $ p $ and $ p' $ are welfare-equivalent if $ p'[\a(\succsim_i,a)]=  p[\a(\succsim_i,a)] $ for every $ i $ and $ a\in \a$. Upper-equivalence only requires that $ p'[\a(\succsim_i,a)]=  p[\a(\succsim_i,a)] $ for every $ i $ and $ a\in \mathcal{A}_i(\gg_i,\o_i)$.

\begin{definition}[Upper-equivalence]
	Two allocations $ p $ and $ p' $ are \textbf{upper-equivalent} if for every $ i $ and $ a\in \mathcal{A}_i(\gg_i,\o_i)$, $ p'[\a_i(\succsim_i,a)]=  p[\a_i(\succsim_i,a)] $.
\end{definition}

\begin{remark}\label{rmk:object:upper-equivalence}
	In the object assignment model, for any preference relation $ \succsim_i $, we define $ O(\gg_i,\emptyset)\coloneqq\{o\in O: \exists o'\in O \text{ such that }o\succ_i o'\succsim_i \emptyset\} $. If agents have strict preferences, then two allocations $ p $ and $ p' $ are upper-equivalent if for every $ i $ and $ o\in  O(\gg_i,\emptyset)$, $ p_{i,o}=p'_{i,o} $.
\end{remark}

\subsection{Mechanism}\label{section:model:mechanism} 

Fixing $ I $, $ \a$ and $ \r $ in our model, a  \textit{mechanism}  is a function $ \psi:\mathcal{R}\rightarrow \Delta(\mathcal{A}) $ that finds an allocation for every $ \succsim_I\in \mathcal{R} $. If $ \psi $ always finds a deterministic allocation, then it is a \textit{deterministic mechanism}. Otherwise, it is a \textit{random mechanism}. We say $ \psi $ is \textit{IR} if, for all $ \succsim_I\in \mathcal{R} $, $ \psi(\succsim_I) $ is IR for $ \succsim_I $.

We extend stochastic dominance and size dominance to mechanisms. For any two mechanisms $ \psi' $ and $ \psi $, $ \psi' $ \textit{weakly stochastically dominates} $ \psi $ if, for all $ \succsim_I\in \mathcal{R} $, $ \psi'(\succsim_I)$ weakly stochastically dominates $\psi(\succsim_I) $. If there further exists some $ \succsim_I $ such that $ \psi'(\succsim_I)$ strictly stochastically dominates $\psi(\succsim_I) $, $ \psi' $ \textit{strictly stochastically dominates} $ \psi $. We say $ \psi' $ and $ \psi $ are \textit{welfare-equivalent} if, for all $ \succsim_I\in \mathcal{R} $, $ \psi'(\succsim_I)$ and $\psi(\succsim_I) $ are welfare-equivalent. Size dominance and size-equivalence are extended similarly.

For every $ \succsim_I\in \mathcal{R} $ and $ i $, let $ \succsim_{-i} $ denote the preference profile of the agents other than $ i $.

\begin{definition}
	A mechanism $ \psi $ is \textbf{strategy-proof} if, for every $ \succsim_I\in \mathcal{R} $, $ i\in I $ and $ \succsim'_i\in \mathcal{R}_i\backslash \{\succsim_i\} $, $ \psi(\succsim_I)\succsim^{sd}_i\psi(\succsim'_i,\succsim_{-i}) $.
\end{definition}

Let $ \psi$ be a strategy-proof mechanism. Given a preference profile $ \succsim_I $, suppose some $ i $ contracts his preferences $ \succsim_i $ at some $ a\in \a_i(\gg_i,\o_i) $, and denote the contraction by $ \succsim'_i $. Then strategy-proofness requires that $ \psi(\succsim_I)\succsim^{sd}_i\psi(\succsim'_i,\succsim_{-i}) $ and $ \psi(\succsim'_i,\succsim_{-i})\succsim'^{sd}_i \psi(\succsim_I) $. Because $ \a_i(\succsim_i ,a)=\a_i(\succsim'_i,a) $, the two stochastic dominance conditions require that $ \psi(\succsim_I)[\a_i(\succsim_i,a)]=\psi(\succsim'_i,\succsim_{-i})[\a_i(\succsim_i,a)] $. In words, the probability mass of the upper contour set of $ a $ is invariant in the allocations found by $ \psi $. This is what we call \textit{contraction-invariance}.

\begin{definition}\label{defn:contraction-invariance}
	A mechanism $ \psi $ is \textbf{contraction-invariant} if, for every $ \succsim_I\in \mathcal{R} $, $ i\in I $ and $ a\in \mathcal{A}_i(\gg_i,\o_i)$, if $ \succsim'_i $ is a contraction of $ \succsim_i $ at $ a $, then $ \psi(\succsim_I)[\a_i(\succsim_i,a)]=\psi(\succsim'_i,\succsim_{-i})[\a_i(\succsim_i,a)] $.
	
	If every $ \r_i $ accommodates truncation, a mechanism $ \psi $ is \textbf{truncation-invariant} if, for every $ \succsim_I \in \mathcal{R}$, $ i\in I $ and $ a\in \mathcal{A}_i(\gg_i,\o_i)  $, $ \psi(\succsim_I)[\a_i(\succsim_i,a)]=\psi(\succsim^a_i,\succsim_{-i})[\a_i(\succsim_i,a)] $.
\end{definition}

The following observation will be used in Theorem \ref{thm1}. Let $ \psi $ be an IR and strategy-proof mechanism. For every $ \succsim_I $ and $ i $, if $ \succsim'_i $ is a contraction of $ \succsim_i $ at $ a\in \mathcal{A}_i(\gg_i,\o_i) $, then $ \a_i(\succsim_i,a)=\a_i(\succsim'_i,\o_i) $. Because $ \psi $ is IR, $ \psi(\succsim'_i,\succsim_{-i})[\a_i]=\psi(\succsim'_i,\succsim_{-i})[\a_i(\succsim'_i,\o_i)] $. So,
\begin{align*}
\psi(\succsim'_i,\succsim_{-i})[\a_i]=\psi(\succsim'_i,\succsim_{-i})[\a_i(\succsim'_i,\o_i)]=\psi(\succsim'_i,\succsim_{-i})[\a_i(\succsim_i,a)]=\psi(\succsim_I)[\a_i(\succsim_i,a)].
\end{align*}
That is, $ i $'s participation size in $ \psi(\succsim'_i,\succsim_{-i}) $ is equal to $ \psi(\succsim_I)[\a_i(\succsim_i,a)]$. It means that, by letting $ i $ contract his preferences at different $ a\in \mathcal{A}_i(\gg_i,\o_i) $, we can recover his welfare in $ \psi(\succsim_I) $ from his participation size in different contracted preference profiles.

A deterministic mechanism $ \psi $ is contraction-invariant if for every $ \succsim_I $, $ i$ and $ a\in \mathcal{A}_i(\gg_i,\o_i) $, if $ \succsim'_i $ is a contraction  of $ \succsim_i $ at $ a $, then $ \psi(\succsim'_i,\succsim_{-i})\succsim'_i a $ if and only if $ \psi(\succsim_I)\succsim_i a $.

\begin{remark}\label{rmk:object:truncation-invariance}
	Truncation-invariance is equivalent to a seemingly stronger requirement: for every $ \succsim_I \in \mathcal{R}$, $ i\in I $ and $ a\in \mathcal{A}_i(\gg_i,\o_i)  $, $ \psi(\succsim_I)[\a_i(\sim_i,b)]=\psi(\succsim^a_i,\succsim_{-i})[\a_i(\sim_i,b)] $ for all $ b\in \a_i(\succsim_i,a) $. To see this, note that for every $ b\in \a_i(\succ_i,a) $, $ \succsim^b_i $ is the truncation of both $ \succsim_i $ and $ \succsim^a_i $ at $ b $. By truncation-invariance, $ \psi(\succsim_I)[\a_i(\succsim_i,b)]=\psi(\succsim^b_i,\succsim_{-i})[\a_i(\succsim_i,b)]=\psi(\succsim^a_i,\succsim_{-i})[\a_i(\succsim_i,b)] $.
	In the object assignment model, when agents have strict preferences, \cite{hashimoto2014two} call this stronger requirement \textbf{weak truncation robustness}.\label{footnote:truncation-invaraince}
\end{remark}

\section{Main results} \label{section:theorems}

Our first result presents the relation between stochastic dominance and size dominance for IR and strategy-proof mechanisms.

\begin{theorem}\label{thm1}
	Let $ \psi $ and $ \psi' $ be two IR and strategy-proof mechanisms.
	\begin{enumerate}
		\item Under NI, $ \psi' $ weakly stochastically dominates $ \psi $ $ \implies $ $ \psi' $ weakly size dominates $ \psi $.
		
		\item Under Richness, $ \psi' $ weakly size dominates $ \psi $ $ \implies $ $ \psi' $ weakly stochastically dominates $ \psi $.
	\end{enumerate} 
\end{theorem}

\begin{proof}
	
	(1) Fix a preference profile $ \succsim_I $. Suppose $ p' $ and $ p $ are two IR allocations and $ p' $ weakly stochastically dominates $ p $. Then for every $ i $, $ p[\a_i]=p[\a_i(\succ_i,\o_i)]+p[\a_i(\sim_i,\o_i)] $ and $ p'[\a_i]=p'[\a_i(\succ_i,\o_i)]+p'[\a_i(\sim_i,\o_i)] $. NI requires $ \a_i(\sim_i,\o_i)=\emptyset $ and weak stochastic dominance requires $ p'[\a_i(\succ_i,\o_i)]\ge p[\a_i(\succ_i,\o_i)] $. So $ p'[\a_i]\ge p[\a_i] $, meaning that $ p' $ weakly size dominates $ p $.

	(2) Suppose $ \psi' $ weakly size dominates $ \psi$, but $ \psi' $ does not weakly stochastically dominate $ \psi $. It means that there exists $ \succsim_I $ and $ i $ such that $ \psi'(\succsim_I) \not\succsim_i^{sd} \psi(\succsim_I)$. Then there must exist $ a\in \a $ such that $ \psi(\succsim_I)[\a(\succsim_i,a)]>\psi'(\succsim_I)[\a(\succsim_i,a)] $.
	Because $ \psi $ and $ \psi' $ are IR, it must be that $ a\succ_i\o_i $ and thus $ \a(\succsim_i,a)=\a_i(\succsim_i,a) $. Because $ \psi' $ weakly size dominates $ \psi$, it must be that $ a\in \a_i(\gg_i,\o_i) $; otherwise $ \psi'(\succsim_I)[\a_i(\succsim_i,a)]=\psi'(\succsim_I)[\a_i]\ge \psi(\succsim_I)[\a_i]= \psi(\succsim_I)[\a_i(\succsim_i,a)] $.

	Under Richness, there exists a contraction $ \succsim'_i $ of $ \succsim_i $ at $ a $. Because $ \psi $ and $ \psi' $ are strategy-proof, they are contraction-invariant. So $ \psi(\succsim'_i,\succsim_{-i})[\a_i]=\psi(\succsim_I)[\a(\succsim_i,a)] $ and $ \psi'(\succsim'_i,\succsim_{-i})[\a_i]=\psi'(\succsim_I)[\a(\succsim_i,a)] $.   This means that $ \psi(\succsim'_i,\succsim_{-i})[\a_i]>\psi'(\succsim'_i,\succsim_{-i})[\a_i] $. But it contradicts the assumption that $ \psi' $ weakly size dominates $ \psi $.
\end{proof}

As explained in Section \ref{section:random:allocation}, under NI, weak stochastic dominance implies weak size dominance for IR allocations. So strategy-proofness plays a role only in the second part of Theorem \ref{thm1}. The insight is that, as explained in Section \ref{section:model:mechanism},  for any IR and strategy-proof mechanism $ \psi $, we can recover every $ i $'s welfare in every $ \psi(\succsim_I) $ from $ i $'s participation size in contracted preference profiles. To see this, note that, first, for every $ a\in \a(\sim_i,\o_i) $, IR implies that $ \psi(\succsim_I)[\a(\succsim_i,a)]=1 $. Second, for every $ a\in \a_i(\gg_i,\o_i) $, by contraction-invariance, $ \psi(\succsim_I)[\a(\succsim_i,a)] $ is equal to $ i $'s participation size in $ \psi(\succsim'_i,\succsim_{-i}) $ where $ \succsim'_i $ is a contraction of $ \succsim_i $ at $ a $. Last, if $ \a_i(\succ_i,\o_i)\neq \a_i(\gg_i,\o_i) $, meaning $ \a_i(\sim_i,\o_i)=\emptyset $, for every $ a\in \a_i(\succ_i,\o_i)\backslash \a_i(\gg_i,\o_i)$, $ \psi(\succsim_I)[\a(\succsim_i,a)]$ is equal to $ i $'s participation size in $\psi(\succsim_I) $. So for IR and strategy-proof mechanisms, size dominance implies stochastic dominance.
    
An immediate corollary to Theorem \ref{thm1} is that, when NI and Richness both hold, stochastic dominance is equivalent to size dominance for IR and strategy-proof mechanisms.

\begin{corollary}\label{corollary:thm1:1}
	Let $ \psi $ and $ \psi' $ be two IR and strategy-proof mechanisms. Under NI and Richness:
	\begin{enumerate}
		\item $ \psi' $ strictly stochastically dominates $ \psi $ $ \Leftrightarrow $ $ \psi' $ strictly size dominates $ \psi $.
		
		\item $ \psi' $ is welfare-equivalent to $ \psi $ $ \Leftrightarrow $ $ \psi' $ is size-equivalent to $ \psi $.\label{prop2.equivalence}
	\end{enumerate}
\end{corollary}

\begin{remark}\label{rmk:thm1:infinite}
	When $ \a $ is infinite, Theorem \ref{thm1} holds as long as the measurability issue is dealt with. Specifically,  for every $ \succsim_i\in \r_i $ and $ a\in \a $, we assume that all of $ \a(\succsim_i,a) $, $ \a(\succ_i,a) $, $ \a(\sim_i,a) $ and the similar subsets defined for $ \a_i $ are measurable. A random allocation is a probability measure so that a random allocation $ p' $ weakly stochastically dominates another random allocation $ p $ for $ \succsim_i $ if $ p'[\a(\succsim_i, a)] \ge p[\a(\succsim_i, a)] $ for all $ a\in \a $. The proof of Theorem \ref{thm1} remains the same as before.\label{footnote:infinite:a}
\end{remark}

Our second result presents another implication of strategy-proofness. It proves that, under Richness, for any two strategy-proof mechanisms, if upper-equivalence implies welfare-equivalence for their outcomes in every preference profile, then the two mechanisms are welfare-equivalent. As we will show in Section \ref{section:application}, this lemma is at the heart of several results of economic relevance in  distinct models.

\begin{lemma}\label{thm2}
	Let $ \psi $ and $ \psi' $ be two strategy-proof mechanisms. Under Richness, if for every $ \succsim_I\in \mathcal{R} $, $ \psi(\succsim_I)$ and $ \psi'(\succsim_I) $ are upper-equivalent$ \implies $$ \psi(\succsim_I)$ and $ \psi'(\succsim_I) $ are welfare-equivalent, then $ \psi $ and $ \psi' $ are welfare-equivalent.
\end{lemma}

\begin{proof}
	Suppose towards a contradiction that $ \psi $ and $ \psi' $ are not welfare-equivalent. Define $ \r^*\coloneqq \{\succsim_I\in \r: \psi(\succsim_I) \text{ and }  \psi'(\succsim_I)  \text{ are not welfare-equivalent}\} $. Then for every $ \succsim_I\in \r^* $, $ \psi(\succsim_I)$ and $ \psi'(\succsim_I) $ are not upper-equivalent. It means that there exist $ i $ and $ a\in \a_i(\gg_i,\o_i) $ such that 
   $\psi(\succsim_I)[\a_i(\succsim_i,a)]\neq \psi'(\succsim_I)[\a_i(\succsim_i,a)]$. Under Richness, let $ \succsim'_i $ be any contraction of $ \succsim_i $ at $ a $. Because $ \psi $ and $ \psi' $ are strategy-proof, $ \psi(\succsim'_i,\succsim_{-i})[\a_i(\succsim'_i,a)]=\psi(\succsim_I)[\a_i(\succsim_i,a)] $ and $ \psi'(\succsim'_i,\succsim_{-i})[\a_i(\succsim'_i,a)]=\psi'(\succsim_I)[\a_i(\succsim_i,a)] $. So $ \psi(\succsim'_i,\succsim_{-i})[\a_i(\succsim'_i,a)]\neq  \psi'(\succsim'_i,\succsim_{-i})[\a_i(\succsim'_i,a)]$, meaning that $ \psi(\succsim'_i,\succsim_{-i}) $ and $ \psi'(\succsim'_i,\succsim_{-i}) $ are not welfare-equivalent. Therefore, $ (\succsim'_i,\succsim_{-i})\in \r^* $. It means that for every $ \succsim_I\in \r^* $, we can find a contracted preference profile $ (\succsim'_i,\succsim_{-i})\in \r^* $. So $ \r^* $ is an infinite set, contradicting the fact that both $ I $ and $ \a $ are finite.
\end{proof}

\begin{remark}\label{rmk:lemma1:formulation}
	An equivalent formulation of Lemma \ref{thm2} is used in Section \ref{section:application}: for any two strategy-proof mechanisms $ \psi $ and $ \psi' $, if for every $ \succsim_I $, $ \psi(\succsim_I)$ and $ \psi'(\succsim_I) $ are not welfare-equivalent$ \implies $$ \psi(\succsim_I)$ and $ \psi'(\succsim_I) $ are not upper-equivalent, then $ \psi $ and $ \psi' $ must be welfare-equivalent.
	
	To prove $ \psi(\succsim_I) $ and $ \psi'(\succsim_I) $ are not upper-equivalent, we need to find $ i\in I $ and $ a\in \a_i(\gg_i,\o_i) $ such that $ \psi(\succsim_I)[\a_i(\succsim_i,a)]\neq \psi'(\succsim_I)[\a_i(\succsim_i,a)] $. If $ \psi $ and $ \psi' $ are deterministic mechanisms, we need to find $ i\in I $ and $ a\in \a_i(\gg_i,\o_i) $ such that $ \psi(\succsim_I)\succsim_i a\succ_i  \psi'(\succsim_I)$ or $ \psi'(\succsim_I)\succsim_i a\succ_i  \psi(\succsim_I)$.
\end{remark}

Theorem \ref{thm1} and Lemma \ref{thm2} both rely on the contraction-invariance property, which is implied by strategy-proofness in the presence of outside options. But they have two differences. 
\textit{First}, Theorem \ref{thm1} does not rely on the finiteness of $ \a $ (see Remark \ref{rmk:thm1:infinite}), but Lemma \ref{thm2} relies on it. With a finite $ \a $, starting with any preference profile $ \succsim_I $, by iteratively contracting agents' preferences, we can find a sequence of preference profiles $ (\succsim^1_I,\succsim^2_I,\ldots,\succsim^n_I) $ such that $ \succsim^1_I=\succsim_I $ and for every $ i $ in $ \succsim^n_I $, $ \a_i(\gg^n_i,\o_i)=\emptyset $. Obviously, the allocations found by any two strategy-proof mechanisms $ \psi $ and $ \psi' $ for every such $ \succsim^n_I $ are upper-equivalent. This is the starting point where we apply the assumption that upper-equivalence  implies welfare-equivalence. By inductively applying this assumption and contraction-invariance, we can conclude that  $ \psi(\succsim_I) $ and $ \psi'(\succsim_I) $ are welfare-equivalent. 
\textit{Second}, Lemma \ref{thm2} holds for strategy-proof mechanisms that are not IR, but Theorem \ref{thm1} requires IR. This is because Theorem \ref{thm1} relates agents' participation sizes to their welfare. As we have explained, an agent's participation size in an allocation does not contain information about his welfare unless IR is assumed.

When agents' preference domains accommodate truncation, our results hold for non-strategy-proof mechanisms that are truncation-invariant. Truncation-invariance can be implied by properties unrelated to incentive. In Section \ref{section:matching}, under weak assumptions the doctor-optimal stable mechanism may not be strategy-proof, but we prove that doctor-optimal stability implies truncation-invariance. Some market design mechanisms are not strategy-proof, but their algorithmic procedures imply truncation-invariance. For example, when agents have strict preferences, the \textit{boston mechanism}  and the \textit{probabilistic serial mechanism} both have the procedure that agents report favorite objects step by step, and the allocation found at each step is finalized and independent of agents' preferences  revealed after the step. This means that both mechanisms are truncation-invariant.

\begin{remark}\label{rmk:truncation-invariant}
	There may be a concern that agents may misreport preferences in non-strategy-proof mechanisms. On the one hand, in some models there are reasons to focus on non-strategy-proof mechanisms, because strategy-proofness conflicts with other goals. For example, in the object assignment model, several papers have demonstrated the conflict between efficiency, fairness, and strategy-proofness \citep{zhou1990conjecture,bogomolnaia2001new,martini2016strategy,nesterov2015fairness}. When efficiency and fairness are primary goals, we have to give up strategy-proofness. On the other hand, we are aware that recent studies suggest that agents may not report true preferences in strategy-proof mechanisms \citep{li2015obviously,ashlagi2015no,rees2018experimental,hassidim2017mechanism}. So whether agents will report true preferences in a specific mechanism is an empirical question. Our results hold for non-strategy-proof mechanisms that are truncation-invariant when incentive is not a concern.
\end{remark}

Finally, we clarify our relation to \cite{alva2017strategy}. AM obtain a theorem in the same model we use, and show many useful applications of their theorem. If we restrict attention to deterministic mechanisms, then stochastic dominance will reduce to Pareto dominance and Theorem \ref{thm1} will reduce to AM's theorem. Theorem \ref{thm1} significantly extends AM's theorem by treating deterministic mechanisms and random mechanisms in a unified way. We point out that our results hold for non-strategy-proof mechanisms that are truncation-invariant. This is useful in some applications (see Section \ref{section:application}). As for preference domain assumptions, our NI is borrowed from AM, but there are subtle differences between our Richness and AM's richness assumption. AM's richness assumes that, for every $ \succsim_i\in \mathcal{R}_i $ and every $ a,b\in \mathcal{A}_i $ with $ b\succ_i a\succsim_i \o_i $, there exists $ \succsim'_i\in \mathcal{R}_i $ such that $ b\succ'_i \o_i \succ'_i a $, and for all $ c\in \mathcal{A}_i $, $c\succsim'_i \o_i \Rightarrow c\succ_i a $. AM's richness is more relaxed than ours. It is enough for their result because in deterministic allocations agents' participation sizes  are degenerate and their preferences over deterministic allocations are complete. Because stochastic dominance is an incomplete relation on random allocations, we have to rely on contraction-invariance to connect the allocations in different preference profiles. This explains the formulation of our Richness.  
There is no counterpart of Lemma \ref{thm2} in AM's paper. So we can unify results that AM's theorem cannot cover. This is also why we use most applications in Section \ref{section:application} to show the usefulness of Lemma \ref{thm2}.

In Section \ref{section:efficiency:frontier} we present useful corollaries to our results.

\subsection{Efficiency frontier of IR and strategy-proof mechanisms}\label{section:efficiency:frontier}

In any class of mechanisms, a member is said to be on the \textit{efficiency frontier} if it is not strictly stochastically dominated by any other member. 
In the object assignment model, \cite{bogomolnaia2001new} call a mechanism \textit{ordinally efficient} if it is not strictly stochastically dominated by any other mechanism. Clearly, in any class of mechanisms its ordinally efficient members must be on the efficiency frontier. But \cite{erdil2014strategy} proves that in the class of IR and strategy-proof mechanisms,  non-wastefulness, which is weaker than ordinal efficiency, is sufficient for a member to be on the efficiency frontier. We can similarly define ordinal efficiency in our abstract model. But below we present some properties weaker than ordinal efficiency that are sufficient for a mechanism to be on the efficiency frontier of IR and strategy-proof mechanisms. In the object assignment model, these properties are weaker than  non-wastefulness (see Section \ref{section:object}).

\begin{definition}[Bidominance]\label{defn:bidominance}
	For any two allocations $ p $ and $ p' $, $ p' $ \textbf{bidominates} $ p $ if $ p' $ strictly stochastically dominates $ p $ and also strictly size dominates $ p $; $ p' $ \textbf{strongly bidominates} $ p $ if $ p' $ bidominates $ p $ and for every $ i $, $ p'\succ^{sd}_i p \Rightarrow p'[\a_i]> p[\a_i] $. 
	
	An allocation $ p $ is \textbf{unbidominated} if it is not bidominated by any other allocation, and is \textbf{not-strongly-bidominated} if it is not strongly bidominated by any other allocation. We call a mechanism unbidominated/not-strongly-bidominated if it always finds such an allocation for every preference profile.
\end{definition}

Example \ref{example:bidominance} illustrates the differences between these concepts. 

\begin{example}[Illustration of bidominance]\label{example:bidominance}
	Suppose $ I=\{i,j\} $ and $\a= \{a,b,c,o^*\} $. The participants in each deterministic allocation are: $ I(a)=I(b)=\{i,j\} $, $ I(c)=\{i\} $, and $ I(o^*)=\emptyset $. So $ i $ participates in $ a,b,c $, and $ j $ participates in $ a,b $. Agents have the following preferences: $ b\succ_i c\succ_i a\succ_i o^* $ and $ b\succ_j a\succ_j c\sim_j o^* $. Consider the following five allocations:
	\begin{table}[!htb]
		\centering
		\begin{tabular}{|c|c|c|c|c|}
			\hline
			$ p^1 $ & $ p^2 $ & $ p^3 $ & $ p^4 $ & $ p^5 $ \\ \hline
			$ 1/2a,1/2o^* $ & $ 1/2b,1/2o^* $ & $ 1/2b,1/2c $ & $ 1b $ & $ 1a $ \\ \hline
		\end{tabular}
	\end{table}

   All of $ p^2,p^3,p^4,p^5 $ strictly stochastically dominate $ p^1 $. $ p^2 $ does not bidominate $ p^1 $ because no agent's participation size is increased in $ p^2 $. $ p^3 $ bidominates $ p^1 $ because $ i $'s participation size is increased in $ p^3 $. But $ p^3 $ does not strongly bidominate $ p^1 $ because $ j $ is better off in $ p^3 $ yet his participation size is not increased in $ p^3 $. $ p^4 $ and $ p^5 $ strongly bidominate $ p^1 $ because every agent's participation size is increased. 
   
   Because every agent most prefers $ b $, $ p^4 $ is the only ordinally efficient allocation. $ p^5 $ is strictly stochastically dominated by $ p^4 $. But because every agent's participation size in $ p^5 $ is one, $ p^5 $ is unbidominated. $ p^3 $ is bidominated by $ p^4 $, because both agents become strictly better off and $ j $'s participation size is increased in $ p^4 $. But $ p^3 $ is not-strongly-bidominated, because the only way to strictly stochastically dominate $ p^3 $ is to move probabilities from $ c $ to $ b $ and this makes $ i $ strictly better off but does not increase $ i $'s participation size.
\end{example}

We first present a corollary to Theorem \ref{thm1}. If an allocation $ p $ is unbidominated, then any allocation $ p' $ that strictly stochastically dominates $ p $ has to be size-equivalent to $ p $.
So if two strategy-proof mechanisms $ \psi' $ and $ \psi'' $ strictly stochastically dominate some IR and unbidominated mechanism $ \psi $, then $ \psi' $ and $ \psi'' $ are respectively size-equivalent to $ \psi $ and therefore are size-equivalent to each other. By Theorem \ref{thm1}, under Richness,  $ \psi' $ and $ \psi'' $ are welfare-equivalent. It means that in terms of welfare, at most one strategy-proof mechanism can strictly stochastically dominate an IR and unbidominated mechanism. Moreover, if $ \psi $ is strategy-proof, then $ \psi' $ and $ \psi'' $ have to be welfare-equivalent to $ \psi $, which means that they cannot strictly stochastically dominate $ \psi $. So in the class of IR and strategy-proof mechanisms, unbidominated mechanisms are on the efficiency frontier.

\begin{corollary}\label{corollary:thm1:2}
	Let $ \psi $ be an IR and unbidominated mechanism. Under Richness:
	\begin{enumerate}
		\item If $ \psi' $ and $ \psi'' $ are two strategy-proof mechanisms that strictly stochastically dominate $ \psi $, then $ \psi' $ and $ \psi'' $ are welfare-equivalent.
		
		\item If $ \psi $ is strategy-proof, then $ \psi $ is not strictly stochastically dominated by any other  strategy-proof mechanism.
	\end{enumerate}
\end{corollary}

We then present a corollary to Lemma \ref{thm2}. If an IR and not-strongly-bidominated allocation $ p $ is strictly stochastically dominated by an allocation $ p' $, then there exists $ i $ such that $ p'\succ^{sd}_i p $ and $ p'[\a_i]\le p[\a_i] $. Suppose NI holds. Then $ p'\succ^{sd}_i p $ implies $ p'[\a_i]\ge p[\a_i]  $. So we must have $ p'[\a_i]= p[\a_i] $. Then $ p'\succ^{sd}_i p $ requires that there exist $ a\in \a_i(\gg_i,\o_i) $ such that $ p'[\a_i(\succsim_i,a)]>p[\a_i(\succsim_i,a)] $. That is, $ p $ and $ p' $ are not upper-equivalent. This observation means that, under NI, if an IR, not-strongly-bidominated and strategy-proof mechanism $ \psi $ is strictly stochastically dominated by another strategy-proof mechanism $ \psi' $, then for every $ \succsim_I $ where $ \psi'(\succsim_I) $ and $ \psi(\succsim_I) $ are not welfare-equivalent, they are not upper-equivalent. Suppose Richness also holds. By Lemma \ref{thm2} and Remark \ref{rmk:lemma1:formulation}, $ \psi $ and $ \psi' $ are welfare-equivalent, a contradiction. So under NI and Richness, in the class of IR and strategy-proof mechanisms, not-strongly-bidominated members are on the efficiency frontier.

\begin{corollary}\label{corollary:thm2:2}
	Under NI and Richness, an IR, not-strongly-bidominated and strategy-proof mechanism is not strictly stochastically dominated by any other strategy-proof mechanism.
\end{corollary}

Because unbidominated mechanisms are not-strongly-bidominated but the converse is not true, Corollary \ref{corollary:thm2:2} is stronger than the second statement of Corollary \ref{corollary:thm1:2} when NI and Richness both hold. However, Corollary \ref{corollary:thm2:2} relies on a finite $ \a $, but Corollary \ref{corollary:thm1:2} does not.

When we are interested in a subclass of IR and strategy-proof mechanisms, we can define unbidominated/not-strongly-bidominated mechanisms within the subclass. Then Corollary \ref{corollary:thm1:2} and Corollary \ref{corollary:thm2:2} hold within the subclass. We will present an application of this result in Section \ref{section:matching}.
 
\begin{remark}\label{rmk:AM:PCPM}
	AM call a deterministic allocation $ a\in \a $ \textit{Pareto-constrained participation-maximal} if there does not exist $ b\in \a $ that \textit{weakly} Pareto dominates $ a $ and strictly participation expands $ a $ (i.e., $ I(a)\subsetneq I(b) $). It is slightly different from unbidominance. But under NI they coincide, because if $ b $ weakly Pareto dominates $ a $ and strictly participation expands $ a $, then $ b $ must strictly Pareto dominate $ a $. There is no counterpart of strong bidominance in AM's paper.
\end{remark}

\section{Applications}\label{section:application}

We present four applications to show the usefulness of our results. AM have presented several applications to show the usefulness of Theorem \ref{thm1} for deterministic mechanisms. In the first application we show its usefulness for random mechanisms in the object assignment model.
We use the remaining applications to show the usefulness of Lemma \ref{thm2}. In the second application we propose the size improvement relation between mechanisms in the object assignment model and use Lemma \ref{thm2} to obtain an impossibility result.  In the third application we use Lemma \ref{thm2} to prove all results of \cite{hirata2017stable} and a characterization result of \cite{combe2020design} in two matching models. In the last application we use Lemma \ref{thm2} to prove main results of \cite{sonmez1999strategy} and \cite{ehlers2018strategy} in a generalized object assignment model.

\subsection{Strategy-proof stochastic dominance in the object assignment model}\label{section:object}

The object assignment model has been defined in Remark \ref{rmk:object:defn} and \ref{rmk:object:randomallocation}. Recall that $ I $ is a finite set of agents, and $ \tilde{O}\coloneqq O\cup \{\emptyset\} $ is a finite set of heterogeneous objects with $ \emptyset $ being the outside option. Assume that every agent's preference domain consists of of all strict preferences over $ \tilde{O} $. Denote this preference domain by $ \mathcal{P} $. So NI and Richness are satisfied. 

An allocation is represented by a nonnegative matrix $ p=(p_{i,o})_{i\in I, o\in \tilde{O}} $ where every $ p_{i,o} $ is the probability share of $ o\in \tilde{O} $ assigned to $ i $. For every $ i $ and $ o $, we define $ \norm{p_i}\coloneqq \sum_{o\in O}p_{i,o} $, $ \norm{p_o}\coloneqq \sum_{i\in I}p_{i,o} $, and $ \norm{p}\coloneqq \sum_{i\in I}\norm{p_i} $. So $ \norm{p_i} $ is $ i $'s \textit{participation size} in $ p $. We call $ \norm{p} $ the \textit{size} of $ p $. \cite{erdil2014strategy} says that a mechanism $ \psi' $ is of \textit{greater size} than another mechanism $ \psi $ if, for every $ \succsim_I $, $ \norm{\psi'(\succsim_I)}\ge \norm{\psi(\succsim_I)} $, and for some $ \succsim'_I $, $ \norm{\psi'(\succsim'_I)}> \norm{\psi(\succsim'_I)} $. If $ \psi' $ strictly size dominates $ \psi $, then $ \psi' $ is of greater size than $ \psi $. But the converse is not true.

An allocation $ p $ is \textit{individually rational} (IR) if, for every $ i\in I $ and $ o\in O $, $ p_{i,o}>0 \Rightarrow o \succ_i \emptyset $.  An allocation $ p $ is \textit{non-wasteful} if it is IR and there do not exist $ i\in I$ and $ o\in O $ such that $ \norm{p_o}<q_o $ and $ o\succ_i o' $ for some $ o'\in \tilde{O} $ with $ p_{i,o'}>0 $. Below we prove that if an allocation is non-wasteful, it is unbidominated. But the converse is not true.

\begin{lemma}
	In the object assignment model with strict preferences, a non-wasteful allocation must be unbidominated, but an unbidominated allocation may be wasteful.
\end{lemma}
\begin{proof}
	Suppose a non-wasteful allocation $ p $ is bidominated by another allocation $ p' $. Then it must be that $ \norm{p'} >\norm{p}$. So there exists $ o\in O $ such that $ \norm{p_o}<\norm{p'_o}\le q_o $, and therefore there exists $ i\in I $ such that $ p_{i,o}<p'_{i,o} $. Because $ p'_i $ weakly stochastically dominated $ p_i $ for $ i $, $ \sum_{o'\succ_i o}p'_{i,o'}\ge \sum_{o'\succ_i o}p_{i,o'} $. Because  $ p_{i,o}<p'_{i,o} $, we have $ \sum_{o'\succsim_i o}p'_{i,o'}> \sum_{o'\succsim_i o}p_{i,o'} $. This means that $\sum_{o'\succsim_i o}p_{i,o'}<1  $, and therefore there must exist $ o'\in \tilde{O} $ such that $ o\succ_i o' $ and $ p_{i,o'}>0 $. However, given $ \norm{p_o}< q_o $, this means that $ p $ is wasteful, which is a contradiction.
	
	On the other hand, we give an example of unbidominated allocations that are wasteful. Suppose there are two agents $ i,j $ and two objects $ o,o' $, each having two copies. Both agents prefer $ o $ to $ o' $ and prefer $ o' $ to $ \emptyset $. The allocation in which both agents obtain $ (1/2o, 1/2o') $ is unbidominated because every agent's participation size is already one. But it is wasteful because $ o $ is wasted.
\end{proof}

Now we are ready to present the implication of Theorem \ref{thm1}. Because agents have strict preferences, if two mechanisms are welfare-equivalent, they coincide.

\begin{proposition}\label{corollary:object:allocation}
	In the object assignment model with strict preferences, for any two IR and strategy-proof mechanisms $ \psi $ and $ \psi' $:
	\begin{enumerate}
		\item $ \psi' $ strictly size dominates $ \psi $ $ \Longleftrightarrow $ $ \psi' $ strictly stochastically dominates $ \psi $;
		
		\item $ \psi' $ is size-equivalent to $ \psi $ $ \Longleftrightarrow $ $ \psi' = \psi $;
		
		\item If $ \psi $ is not-strongly-bidominated, $ \psi $ is not strictly stochastically dominated by $ \psi' $.
	\end{enumerate}
\end{proposition}

We obtain Erdil's following results as corollaries.

\begin{corollary}[\cite{erdil2014strategy}]\label{corollary:erdil}
	In the object assignment model with strict preferences:
	\begin{enumerate}
		\item If a strategy-proof mechanism is non-wasteful, it is not strictly stochastically dominated by any other strategy-proof mechanism;
		
		\item If a strategy-proof mechanism $ \psi' $ strictly stochastically dominates another strategy-proof mechanism $ \psi $, then $ \psi' $ is of greater size than $ \psi $.
	\end{enumerate}
\end{corollary}

Because agents' preference domains accommodate truncation, Proposition \ref{corollary:object:allocation} actually holds for IR and truncation-invariant mechanisms.  For example, because the probabilistic serial mechanism (PS) and the boston mechanism (BM) are truncation-invariant, Proposition \ref{corollary:object:allocation} implies that a non-wasteful and strategy-proof mechanism is not strictly stochastically dominated by PS, and a non-wasteful and strategy-proof deterministic mechanism (such as DA) is not strictly Pareto dominated by BM.\footnote{Here we regard DA and BM as mechanisms to solve the object assignment model, and regard priorities used by objects to rank agents as elements of these mechanisms.} These implications cannot be obtained from Erdil's results.

\begin{remark}\label{rmk:erdil:RP}
	\cite{erdil2014strategy} also shows that the random priority mechanism (RP) is wasteful, but any strategy-proof mechanism that strictly stochastically dominates RP is still wasteful. So on the efficiency frontier of IR and strategy-proof mechanisms, some mechanisms have to be wasteful.  \cite{zhangsizeinvariance} shows that this is not due to RP as a special benchmark: any strategy-proof mechanism that strictly stochastically dominates a mechanism of the same efficiency and fairness properties as RP must be strongly bidominated (and therefore wasteful).
\end{remark}

\subsection{Size improvement in the object assignment model}\label{section:size:improve}

We consider the following question in this subsection. Suppose a mechanism has been used in an object assignment market, but we want to replace it with a new one to assign more acceptable objects to agents. We ask whether the replacement is possible if the two mechanisms satisfy desirable properties, and agents' welfare in the old mechanism needs to be respected (at least in a weak sense). 

We are motivated by markets such as kidney exchange and refugee resettlement. At least two features in these markets capture our attention. First, not all resources in these markets are acceptable to agents. A donated kidney may be medically incompatible with a patient. A refugee may be unable to live in a country/city for various reasons. Second, being unmatched is very undesirable in these markets. It is more important for an agent to get an acceptable resource than to get a better resource when an acceptable resource has been guaranteed. This explains why assigning more acceptable objects to agents can be an important goal in these markets.\footnote{Optimization and market design techniques have been used in kidney exchange and refugee resettlement to increase the number of acceptable assignments. See \cite{roth2005pairwise,sonmez2017incentivized,kratz2019pairwise,andersson2016assigning}.} 

In our question, if the new mechanism Pareto improves the old, all agents will be glad in the replacement. But if some agents are harmed, they may complain and oppose the replacement.  We assume that agents understand that assigning more objects to agents is a desirable goal for the market. But because being unmatched is very undesirable, an agent will complain if and only if he is matched in the old mechanism but becomes unmatched in the new one. In this case we assume that the only excuse we can use to justify the replacement is that strictly more objects are assigned to agents in the new mechanism. In other words, if the new mechanism does not assign strictly more objects, every matched agent in the old mechanism has to be matched in the new one. 

Formally, we consider the object assignment model, allow for weak preferences, and assume that agents are never indifferent between real objects and the outside option of being unmatched. Let $ \mathcal{B} $ denote the set of complete and transitive preference relation $ \succsim $ on $\tilde{O}  $ such that $ \{o\in O:o\sim \emptyset\}=\emptyset $. $ \mathcal{B} $ is the preference domain of all agents. So NI and Richness are satisfied. Size improvement is defined as follows.

\begin{definition}
	A new mechanism $ \psi' $ \textit{size improves} an old mechanism $ \psi $ if
	\begin{enumerate}
		\item $\forall \succsim_I$, $ \norm{\psi'(\succsim_I)}\ge \norm{\psi(\succsim_I)} $, and $\exists \succsim'_I $, $ \norm{\psi'(\succsim'_I)}>\norm{\psi(\succsim'_I)} $;
		
		\item $\forall \succsim_I$, $ \norm{\psi'(\succsim_I)}=\norm{\psi(\succsim_I)} \implies \psi'(\succsim_I)$ is size-equivalent to $ \psi(\succsim_I) $.
	\end{enumerate} 
\end{definition}

The above definition is applicable to deterministic mechanisms as well as to random mechanisms. In random mechanisms agents compare their participation sizes.

Proposition \ref{prop:sizeimprove} is our main result in this subsection. 

\begin{proposition}\label{prop:sizeimprove}
	In the object assignment model with weak preferences, an IR, unbidominated, and truncation-invariant mechanism is not size improved by any other IR and truncation-invariant mechanism.
\end{proposition}

Because strict size dominance implies size improvement, but the converse is not true, Proposition \ref{prop:sizeimprove} is not implied by Theorem \ref{thm1}. A crucial step in the proof of  Proposition \ref{prop:sizeimprove} is to prove the following Lemma \ref{lemma:random}. It enables us to apply Lemma \ref{thm2} to obtain Proposition \ref{prop:sizeimprove}.
	
	\begin{lemma}\label{lemma:random}
		Give any $ \succsim_I $, for any two IR allocations $ p $ and $ p' $, if $ p $ is unbidominated and $ \norm{p'}>\norm{p} $, then $ p $ and $ p' $ are not upper-equivalent. 
	\end{lemma}
\begin{proof}
	The proof is in Appendix \ref{appendix:proof:size-improve}.
\end{proof}
	
\begin{proof}[Proof of Proposition \ref{prop:sizeimprove}]
	Suppose $ \psi $ and $ \psi' $ are two IR and truncation-invariant mechanisms, $ \psi $ is unbidominated, and $ \psi' $ size improves $ \psi $. For every $ \succsim_I $, if $ \psi(\succsim_I) $ and $ \psi'(\succsim_I) $ are not size-equivalent, then it must be that $ \norm{\psi'(\succsim_I)}>\norm{\psi(\succsim_I)} $. By Lemma \ref{lemma:random}, $ \psi'(\succsim_I) $ and $ \psi(\succsim_I) $ are not upper-equivalent. Because NI holds, if $ \psi(\succsim_I) $ and $ \psi'(\succsim_I) $ are not size-equivalent, they are not welfare-equivalent. So  $ \psi(\succsim_I) $ and $ \psi'(\succsim_I) $ are not welfare-equivalent only if they are not upper-equivalent. By Lemma \ref{thm2} (and Remark \ref{rmk:lemma1:formulation}), $ \psi' $ and $ \psi $ must be welfare-equivalent. So they are size-equivalent, which is a contradiction.
\end{proof}

The implication of Proposition \ref{prop:sizeimprove} is as follows. Suppose the old mechanism is not too inefficient, and the two mechanisms satisfy properties that imply truncation-invariance. If it is important to respect agents' welfare in the old mechanism, then the replacement is not implementable, while if a market designer implements the replacement, then it can be inevitable to harm some agents without helping more of the others.

Proposition \ref{prop:sizeimprove} also holds when agents' preferences are strict. When the preference domain is $ \mathcal{P} $, the proof of Lemma \ref{lemma:random} can be simplified. In the following we apply Proposition \ref{prop:sizeimprove} to several well-studied deterministic mechanisms in the literature: deferred acceptance (DA), boston mechanism (BM), top trading cycle (TTC), serial dictatorship (SD), hierarchical exchange (HE) of \cite{papai2000strategyproof}, and trading cycles (TC) of \cite{pycia2017incentive}. Some of these mechanisms involve priorities, which we regard as elements of these mechanisms. All of these mechanisms except BM are strategy-proof.\footnote{\cite{pycia2017school} extend TC and HE to the environment where each object may have several copies, and prove that the extended mechanisms are Pareto efficient and strategy-proof.} As we have explained, the algorithmic procedure of BM implies truncation-invariance. All of these mechanisms except DA is Pareto efficient. DA is unbidominated because it is non-wasteful. So we obtain the following corollary to Proposition \ref{prop:sizeimprove}.

\begin{corollary}
	In the object assignment model with strict preferences, none of DA, BM, TTC, SD, HE, and TC is size improved by any other truncation-invariant mechanism. In particular, they do not size improve each other.
\end{corollary}

When agents have strict preferences, because PS is IR, non-wasteful and truncation-invariant, Proposition \ref{prop:sizeimprove} implies that PS is not size improved by any other truncation-invariant mechanism. In Appendix \ref{appendix:size:improve} we construct a mechanism that strictly size dominates PS. The mechanism is not truncation-invariant, but it is ordinally efficient and envy-free, which are the flagship properties of PS. This answers an open question raised by \cite{huang2017guaranteed} on the existence of an ordinally efficient and envy-free mechanism that improves the size of PS. This also shows that truncation-invariance of the improving mechanism is necessary for Proposition \ref{prop:sizeimprove} to hold. The other parts of Proposition \ref{prop:sizeimprove} are also tight. In Appendix \ref{appendix:size:improve} we also construct an IR and unbidominated mechanism that is not truncation-invariant, but is strictly size dominated by PS. So truncation-invariance of the old mechanism is necessary for Proposition \ref{prop:sizeimprove} to hold. \cite{erdil2014strategy} has constructed a strategy-proof mechanism to bidominate RP. So unbidominance of the old mechanism is also necessary for Proposition \ref{prop:sizeimprove} to hold.

\subsection{Two-sided matching}\label{section:matching}

In this subsection we use Lemma \ref{thm2} to unify the results in two matching models. Mechanisms in these models are deterministic. In the \textbf{first} model known as  many-to-one matching with contracts, a lot of attention in the literature has been paid to the existence and strategy-proofness of stable mechanisms. It is well-known that substitutability, law of aggregate demand (LAD), and irrelevance of rejected contracts (IRC) of hospitals' choice functions are sufficient and in some sense necessary for the existence and strategy-proofness of the doctor-optimal stable mechanism (DOSM), which finds the Pareto optimal stable matching for doctors \citep{hatfield2005matching,aygun2013matching}. Actually, under these assumptions DOSM is the only stable and strategy-proof mechanism. However, if only substitutability and IRC are assumed, DOSM exists but may not be strategy-proof, and if only IRC is assumed, stable mechanisms may not exist.

Surprisingly, \cite{hirata2017stable} (HK for short) prove that the uniqueness of stable and strategy-proof mechanisms and the special role of DOSM hold more fundamentally than what the literature believed before. Only assuming IRC, HK prove that there is at most one stable and strategy-proof mechanism, and if DOSM exists, it is the only candidate. HK also propose a notion of \textit{non-wastefulness} (different than the notion in the object assignment model), and prove that a non-wasteful and strategy-proof mechanism is not strictly Pareto dominated by any other individually rational (for both doctors and hospitals) and strategy-proof mechanism.  Under IRC, stable mechanisms are non-wasteful. So a stable and strategy-proof mechanism is not strictly Pareto dominated by any other individually rational and strategy-proof mechanism.

\begin{proposition}[\cite{hirata2017stable}]\label{prop:HK}
	In the many-to-one matching with contracts model, under IRC:
	\begin{itemize}
		\item There is at most one stable and strategy-proof mechanism, and if DOSM exists, it is the only candidate; 
		
		\item A non-wasteful and strategy-proof mechanism is not strictly Pareto dominated by any other individually rational and strategy-proof mechanism.
	\end{itemize}
\end{proposition}

To explain how HK's results are implied by Lemma \ref{thm2} and actually can be extended to truncation-invariant mechanisms, we need to briefly define his model. The model consists of a finite set of doctors $ D $, a finite set of hospitals $ H $, and a finite set of contracts $ X $. Each contract $ x\in X $ involves a doctor $ d(x) $ and a hospital $ h(x) $. A subset $ X'\subset X$ is a \textit{matching} if it includes at most one contract for each doctor; among $ X' $ denote by $ x(d,X') $ the contract involving $ d $ and by $ x(h,X') $ the set of contracts involving $ h $. Each $ d\in D$ has a strict preference relation $ \succsim_d $ over $ \{x\in X:d(x)=d\}\cup \{\emptyset\} $, where $ \emptyset $ is a null contract. All strict preferences are possible. Each $ h\in H $ has a fixed choice function $ C_h $: for all $ X'\subset X $, $ C_h(X')=C_h(x(h,X'))\subset x(h,X') $.  $ C_h $ satisfies \textit{IRC} if for all $ X'\subset X $ and $ x\in X $, $ x\notin C_h(X'\cup \{x\}) $ implies $ C_h(X'\cup \{x\})=C_h(X') $.
A matching $ X' $ is \textit{individually rational} if $ x(d,X')\succsim_d \emptyset $ for all $ d $ and $ C_h(X')=x(h,X') $ for all $ h $. $ X' $ is \textit{blocked} by $ (h,X'') $ with $ h\in H $ and $ X''\subset X$ if $ X''\cap X'=\emptyset $, $ x(h,X'')\subset C_h(X'\cup X'')$ and $ x(d,X'')\succ_d x(d,X') $ for all $ d\in \{d(x)\}_{x\in X''} $. $ X' $ is \textit{stable} if it is individually rational and unblocked. $ X' $ is \textit{non-wasteful} if it is individually rational and there is no another individually rational matching $ X'' $ such that $ X'\subsetneq X'' $. A mechanism is stable or individually rational if it finds a stable or individually rational matching for every preference profile of doctors $ \succsim_D=(\succsim_d)_{d\in D} $.

HK have proved the following Lemma \ref{lemma:matching:upper-equivalent}. In our terminology, the lemma says that if two stable mechanisms find different matchings for a preference profile of doctors, the matchings are not upper-equivalent. By Lemma \ref{thm2} (and Remark \ref{rmk:lemma1:formulation}), it implies that any two stable and strategy-proof/truncation-invariant mechanisms must coincide.

\begin{lemma}[\cite{hirata2017stable}]\label{lemma:matching:upper-equivalent}
 Under IRC, if there exist two stable mechanisms $ \psi $ and $ \psi' $, then for every preference profile $ \succsim_D $ where $ \psi(\succsim_D)\neq \psi'(\succsim_D) $, there exists $ d\in D $ such that $ \emptyset \neq x(d,X')\neq x(d,X'')\neq \emptyset $.
\end{lemma}

Lemma \ref{lemma:matching:truncation-invariance} proves that if DOSM exists, although it may not be strategy-proof, doctor-optimal stability implies that it must be truncation-invariant. Because Lemma \ref{thm2} holds for truncation-invariant mechanisms, if a stable and strategy-proof/truncation-invariant mechanism exists, it must coincide with DOSM.

\begin{lemma}\label{lemma:matching:truncation-invariance}
	Under IRC, if DOSM exists, it is truncation-invariant.
\end{lemma}

\begin{proof}
	Let $ \psi $ denote the doctor-optimal stable mechanism. For all $ \succsim_D $, all $ d\in D $ and all $ x\in X $ with $ d(x)=d $ such that $ x\succ_d \emptyset $:
	
	1. If $ x\succ_d \psi_d(\succsim_D)$, we prove that $ \psi_d(\succsim^x_d,\succsim_{-d})=\emptyset $, where $ \succsim^x_d $ is the truncation of $ \succsim_d $ at $ x $. Suppose $ \psi_d(\succsim^x_d,\succsim_{-d})\neq\emptyset $. Since $ \psi $ is individually rational, $ \psi_d(\succsim^x_d,\succsim_{-d})\succsim_d x\succ_d \psi_d(\succsim_D) $. Since $ \psi(\succsim^x_d,\succsim_{-d}) $ is stable under $ (\succsim^x_d,\succsim_{-d}) $ and $ \psi_d(\succsim^x_d,\succsim_{-d})\succsim_d x $, $ \psi(\succsim^x_d,\succsim_{-d}) $ must be stable under $ \succsim_D $. But this contradicts the doctor-optimal stability of $ \psi(\succsim_D) $ under $ \succsim_D $ because $ \psi(\succsim_D) $ does not Pareto dominate $ \psi(\succsim^x_d,\succsim_{-d}) $.
	
	2. If $ \psi_d(\succsim_D)\succsim_d x $, we prove that $ \psi_d(\succsim^{x}_d,\succsim_{-d})\succsim^{x}_d x$. Since $ \psi(\succsim_D) $ is stable under $ \succsim_D $ and $ \psi_d(\succsim_D)\succsim_d x $, it must be stable under $ (\succsim^{x}_d,\succsim_{-d}) $. By the doctor-optimal stability of $ \psi(\succsim^{x}_d,\succsim_{-d}) $ under $ (\succsim^{x}_d,\succsim_{-d}) $, $ \psi_d(\succsim^{x}_d,\succsim_{-d})\succsim^{x}_d \psi_d(\succsim_D)$. So $ \psi_d(\succsim^{x}_d,\succsim_{-d})\succsim^{x}_d x$.
\end{proof}

Lemma \ref{lemma:matching:upper-equivalent} and Lemma \ref{lemma:matching:truncation-invariance} together prove the first part of Proposition \ref{prop:HK}. 

Finally, we prove that non-wasteful matchings called by HK are exactly not-strongly-bidominated matchings within the set of individually rational matchings. 

\begin{lemma}\label{lemma:matching:nonwasteful}
	A matching $ X' $ is non-wasteful as called by \cite{hirata2017stable} if and only if $ X' $ is not-strongly-bidominated among individually rational matchings.
\end{lemma}

\begin{proof}
	For any two individually rational matchings $ X' $ and $ X'' $, we prove that $ X' $ is strongly bidominated by $ X'' $ if and only if $ X'\subsetneq X''  $. $ X' $ is strongly bidominated by $ X'' $ if $ X'' $ strictly Pareto dominates $ X' $ and for all $ d\in D $, $ x(d,X'')\succ_dx(d,X') $ implies $ x(d,X') =\emptyset $. So for all $ d\in D $ with $ x(d,X')\neq \emptyset $, $ x(d,X'')=x(d,X') $. It means that $ X'\subsetneq X'' $.
	
	 Conversely, if $ X'\subsetneq X'' $, for all $ d\in D $ with $ x(d,X')\neq \emptyset $, $ x(d,X'')=x(d,X') $ since each doctor signs at most one contract. For all $ d\in D $ with $ x(d,X')=\emptyset \neq x(d,X'') $, individual rationality implies $ x(d,X'')\succ_d x(d,X') $. So $ X'' $ strongly bidominates $ X' $.
\end{proof}

With Lemma \ref{lemma:matching:nonwasteful}, Corollary \ref{corollary:thm2:2} directly implies the second part of Proposition \ref{prop:HK}.

\bigskip

The \textbf{second} matching model we consider is the teacher reassignment model proposed by \cite{combe2020design}. In the model every teacher is initially matched to a school. The question is how to rematch teachers and schools to improve their welfare, given they have strict preferences over each other. Combe et al. propose a class of  mechanisms called \textit{Block Exchange} (BE) to find the set of desirable matchings. They recommend a subclass of BE called \textit{Teacher Optimal BE} (TO-BE). Every TO-BE is strategy-proof for teachers.  In the one-to-one environment, TO-BE is single-valued, and Combe et al. prove that TO-BE is the only strategy-proof mechanism in the class of BE.

\begin{proposition}[\cite{combe2020design}]\label{prop:combe}
	In the one-to-one environment of the teacher reassignment model, TO-BE is the unique strategy-proof mechanism in the class of BE.
\end{proposition}

We will not go into details about the model and the BE mechanisms. Briefly, the one-to-one environment includes equal numbers of teachers $ T $ and schools $ S $. Every $ i\in S\cup T $ has strict preferences $ \succsim_i $ over agents on the other side. A matching is a one-to-one mapping $ \mu:T\rightarrow S $. Let $ \mu_0 $ denote the initial matching. This environment can be accommodated by our abstract model by letting every teacher's outside option be his initial assignment. BE is a generalization of TTC by taking preferences of both sides into account. Combe et al. prove several lemmas to obtain Proposition \ref{prop:combe}. One of their lemmas, called Lemma \ref{lemma:combe} below, proves that if any BE and TO-BE find different matchings for a preference profile of teachers,   the matchings are not upper-equivalent (in our terminology). Then Lemma \ref{thm2} directly implies Proposition \ref{prop:combe}. Combe et al. need other lemmas and use a logic different from Lemma \ref{thm2} to obtain Proposition \ref{prop:combe}.

\begin{lemma}[\cite{combe2020design}]\label{lemma:combe}
	In the one-to-one environment of the teacher reassignment model, let $ \psi $ be any selection of BE. For any preference profile of teachers $ \succsim_{T} $, if $ x =$  TO-BE$ (\succsim_T) \neq \psi(\succsim_T)=y $, then there exists $ t\in T $ such that $x(t)\succ_t y(t)\succ_t \mu_0(t) $.
\end{lemma}

\subsection{Strategy-proofness and single-valued core}\label{section:sonmez_model}

In this subsection we show the logic in Lemma \ref{thm2} also lies behind the main theorems of \cite{sonmez1999strategy} and \cite{ehlers2018strategy}. In a generalized object assignment model, S\"{o}nmez proves that if there exists an IR, Pareto efficient and strategy-proof mechanism, then when the core is nonempty, the core is essentially single-valued and the mechanism selects a core allocation.  \cite{ehlers2018strategy} extends S\"{o}nmez's result to the individually-rational-core (IR-core), which is a superset of the core. If we use $ C_{ir}(\succsim_I) $ to denote the IR-core in every preference profile $ \succsim_I $, then Ehlers' theorem is stated as follows.

\begin{proposition}[\cite{sonmez1999strategy,ehlers2018strategy}]\label{prop:singlevalued:core}
	In S\"{o}nmez's model, if there exists an IR, Pareto efficient and strategy-proof deterministic mechanism $ \psi $, then: 
	\begin{enumerate}
		\item for all $ \succsim_I $ and all $ a,b\in C_{ir}(\succsim_I) $, $ a\sim_i b $ for all $ i\in I $;
		
		\item for all $ \succsim_I $ with $ C_{ir}(\succsim_I)\neq \emptyset $, $ \psi(\succsim_I)\in C_{ir}(\succsim_I) $.
	\end{enumerate}
\end{proposition}

S\"{o}nmez's model is a special case of ours. Specifically, in S\"{o}nmez's model every agent $ i $ is endowed with a finite set of indivisible objects $ \w(i) $, which we interpret as $ i $'s outside option. A deterministic allocation $ a $ assigns every object to an agent. We use $ a(i) $ to denote the set of objects assigned to $ i $. S\"{o}nmez makes two assumptions on agents' preference domains, which are stronger than our NI and Richness:
\begin{itemize}
	\item Assumption A: for all $ \succsim_i\in \mathcal{R}_i $ and all $ a\in \mathcal{A}$, $ a \sim_i \w \Leftrightarrow a(i)=\w(i)$.
	
	\item Assumption B: for all $ \succsim_i\in \mathcal{R}_i $ and all $  a\in \mathcal{A}(\succ_i,\w)
	$, there exists $ \succsim'_i\in \mathcal{R}_i $ such that $ \mathcal{A}(\succsim_i,a)= \mathcal{A}(\succ'_i,\w)=\mathcal{A}(\succsim'_i,a) $ and $\mathcal{A}(\succ_i, a)=\mathcal{A}(\succ'_i, a) $.
\end{itemize}

The two core notions are defined as follows. For all $ T\subset I $, define $ \w(T)\coloneqq\cup_{i\in T}\w(i) $. An allocation $ a $ is said to weakly dominate an allocation $ b $ via the coalition $ T\subset I $ under $ \succsim_I $, denoted by $ a \ wdom_T \ b $, if  $ a(i)\subset \w(T) $ for all $ i\in T $, $ a\succsim_i b $ for all $ i\in T $, and $ a\succ_j b $ for some $ j\in T $. The \textit{core} under $ \succsim_I $ consists of all undominated allocations. The \textit{IR-core} under $ \succsim_I $ is $ C_{ir}(\succsim_I)\coloneqq\{b\in \mathcal{A}:\nexists T \textit{ and }a\in \mathcal{A} \text{ with }a \ wdom_T \ b\text{ and }a(i)=\w(i) \text{ for all }i\in I\backslash T\} $. 
When the IR-core is nonempty, every IR-core allocation is IR and Pareto efficient.

Our first observation is that, when the IR-core is nonempty, it has the contraction-invariance property in the sense of Lemma \ref{lemma3}.

\begin{lemma}\label{lemma3}
	For every $ \succsim_I\in \r$ with $ C_{ir}(\succsim_I)\neq \emptyset $, every $ a\in C_{ir}(\succsim_I) $, every $ i $, and every contraction $ \succsim'_i $ of $ \succsim_i $ at $ a $ that satisfies S\"{o}nmez's Assumption B, $ a\in C_{ir}(\succsim'_i,\succsim_{-i}) $.
\end{lemma}

\begin{proof}
	Suppose $ a\notin C_{ir}(\succsim'_i,\succsim_{-i}) $. Then under $ (\succsim'_i,\succsim_{-i}) $, there exist $ T\subset I $ and $ c\in \mathcal{A} $ such that $ c \ wdom_T \ a $ and $ c(k)=\w(k) $ for all $ k\in I\backslash T $. So $ c\succsim_j a $ for all $ j\in T\backslash \{i\} $, $ c\succsim'_i a $, and either $ c\succ_j a $ for some $ j\in T\backslash \{i\} $, or $ c\succ'_i a $. By Assumption B, $ c\succsim'_i a $ implies $ c\succsim_i a $ and $ c\succ'_i a $ implies $ c\succ_i a $. So under $ \succsim_I $, we must have $ c \ wdom_T \ a $ and $ c(k)=\w(k) $ for all $ k\in I\backslash T $. So $ a\notin C_{ir}(\succsim_I) $, which is a contradiction.
\end{proof}

Suppose there exists an IR, Pareto efficient and strategy-proof deterministic mechanism $ \psi $. If for all $ \succsim_I$, $ C_{ir}(\succsim_I)= \emptyset $, then Proposition \ref{prop:singlevalued:core} holds trivially. Otherwise, we define $ \mathcal{R}^*\coloneqq\{\succsim_I\in \mathcal{R}: C_{ir}(\succsim_I)\neq \emptyset\} $, which is nonempty. Arbitrarily choose $ \succsim^*_I\in  \mathcal{R}^*$ and $ a\in C_{ir}(\succsim^*_I) $. Then we construct a mechanism $ \psi' $ according to the following steps:
\begin{enumerate}
	\item Let $ \psi'(\succsim^*_I)=a $. So $ \psi'(\succsim^*_I)\in C_{ir}(\succsim^*_I) $.
	
	\item For every $ i $ such that $ a\succ^*_i c\succ^*_i \w $ for some $ c\in \a $, let $ \succsim'_i $ be any contraction of $ \succsim^*_i $ at $ a $ satisfying  Assumption B. Let $ \psi'(\succsim'_i,\succsim^*_{-i})=a $. By Lemma \ref{lemma3}, $ \psi'(\succsim'_i,\succsim^*_{-i})\in C_{ir}(\succsim'_i,\succsim^*_{-i}) $.
	
	\item For every $ \succsim_I $ considered in step 2 and every $ i$ such that $ a \succ_i c \succ_i \w $ for some $ c\in \a $, let $ \succsim'_i $ be any contraction of $ \succsim_i $ at $ a $ satisfying  Assumption B. Let $ \psi'(\succsim'_i,\succsim_{-i})=a $. By Lemma \ref{lemma3}, $ \psi'(\succsim'_i,\succsim_{-i})\in C_{ir}(\succsim'_i,\succsim_{-i}) $.
	
	\item Repeat the above step until we cannot find a new preference profile. Since in every step we contract a preference profile in the previous step, we must stop in finite steps. Denote by $ \mathcal{R}' $ the set of preference profiles (including $ \succsim^*_I $) considered in these steps. We know that for all $ \succsim_I\in \mathcal{R}' $, $ \psi'(\succsim_I)=a\in C_{ir}(\succsim_I) $.
	
	\item For all $ \succsim_I\in \mathcal{R}\backslash \mathcal{R}' $, let $ \psi'(\succsim_I)=\psi(\succsim_I) $. 
\end{enumerate}

So in every preference profile, $ \psi' $ either selects an IR-core allocation or selects the outcome of $ \psi $. Thus, $ \psi' $ is IR and Pareto efficient. Lemma \ref{lemma1} proves that if $ \psi' $ and $ \psi $ find different allocations for a preference profile, the allocations must not be upper-equivalent.

\begin{lemma}\label{lemma1}
	For every $ \succsim_I\in \r$ with $ \psi'(\succsim_I) \neq \psi(\succsim_I) $, there exists $ i$ such that $ \psi'(\succsim_I)\succ_i\psi(\succsim_I)\succ_i \w $.
\end{lemma}

\begin{proof}
	For every $ \succsim_I\in \mathcal{R} $ such that $ \psi'(\succsim_I) \neq \psi(\succsim_I) $, it must be that $ \psi'(\succsim_I)\in C_{ir}(\succsim_I) $. Since both $ \psi'(\succsim_I) $ and $ \psi(\succsim_I) $ are IR and Pareto efficient, there must exist distinct $ i,j $ such that $ \psi(\succsim_I)\succ_i \psi'(\succsim_I) $ and $ \psi'(\succsim_I)\succ_j \psi(\succsim_I) $. Suppose for all $ k\in I'\coloneqq\{i\in I: \psi'(\succsim_I)\succ_i \psi(\succsim_I)\} $, $ \psi(\succsim_I)\sim_k \w $. By  Assumption A, $ \psi_k(\succsim_I)= \w(k) $ for all $ k\in I' $. Then $ \psi(\succsim_I)$ can weakly block $ \psi'(\succsim_I) $ via  $ I\backslash I' $ under $ \succsim_I $, since all agents in $ I\backslash I' $ weakly prefer $ \psi(\succsim_I)$ to $ \psi'(\succsim_I)$, and $ i\in I\backslash I' $ strictly prefers $ \psi(\succsim_I)$ to $ \psi'(\succsim_I)$. But it contradicts the fact that $ \psi'(\succsim_I)\in C_{ir}(\succsim_I) $. So there must exist $ k\in I' $ such that $ \psi'(\succsim_I)\succ_k \psi(\succsim_I)\succ_k \w $.
\end{proof}

Now we are ready to use the logic in Lemma \ref{thm2} to prove Proposition \ref{prop:singlevalued:core}.
\begin{proof}[Proof of Proposition \ref{prop:singlevalued:core}]
	For every $ \succsim_I\in \mathcal{R} $ with $ \psi'(\succsim_I) \neq \psi(\succsim_I) $, $ \psi'(\succsim_I)=a\in C_{ir}(\succsim_I) $. By Lemma \ref{lemma1}, there exists $ i $ such that $ a\succ_i\psi(\succsim_I)\succ_i \w $. Let $ \succsim'_i $ be any contraction of $ \succsim_i $ at $ a $ satisfying Assumption B. So $ \psi'(\succsim'_i,\succsim_{-i})=a\succ'_i \w $. Strategy-proofness of $ \psi $ requires that $ \psi(\succsim'_i,\succsim_{-i})\sim'_i \w$. So $ \psi'(\succsim'_i,\succsim_{-i}) \neq \psi(\succsim'_i,\succsim_{-i}) $. Then, as in the proof of Lemma \ref{thm2}, this means that there exist infinite preference profiles, which is a contradiction. So it must be that $ \psi = \psi' $. Note that $ \psi' $ can be constructed for every $ \succsim^*_I\in  \mathcal{R}^*$ and every $ a\in C_{ir}(\succsim^*_I) $. So all allocations in $ C_{ir}(\succsim^*_I) $ are welfare-equivalent to $ \psi(\succsim^*_I) $. It means that $ C_{ir}(\succsim^*_I) $ is essentially single-valued and $ \psi $ selects an IR-core allocation when the IR-core is nonempty.
\end{proof}

\section{Conclusion}\label{section:conclusion}

In allocation problems where agents have outside options, there is a connection between the allocations found by a strategy-proof mechanism when agents vary the ranking of outside options in their preferences. We call this connection the contraction-invariance property. Relying on this property, we obtain two simple but useful results for strategy-proof mechanisms. The two results uncover a unified logic behind several existing results in different market design models. Our abstract model is general enough to accommodate many other economic environments. So it will be interesting to find more applications of our results in other economic models. This is left for future research.


\appendix

\section{Proof of Lemma \ref{lemma:random}}\label{appendix:proof:size-improve}

We first define some notations. For each $ o\in \tilde{O}\coloneqq O\cup \{\emptyset\} $ and each $ \succsim_i \in \mathcal{B} $, we define $ O(\succsim_i,o)\coloneqq \{o'\in O:o'\succsim_i o\} $; $ O(\succ_i,o) $ and $ O(\sim_i,o) $ are similar. We define $ O(\gg_i,\emptyset)\coloneqq \{o\in O:o\succ_i o'\succ_i \emptyset \text{ for some }o'\in O\} $.
Let $ O^1(\succsim_i),O^2(\succsim_i),\ldots, O^{K_i}(\succsim_i) $ be the partition of $ O(\succ_i,\emptyset) $ induced by $ \sim_i $. So each $ O^k(\succsim_i) $ is an indifference class for $ i $. We make the convention that the objects in $ O^k(\succsim_i) $ are better than those in $ O^{k+1}(\succsim_i) $. We suppress the dependence of $ K_i $ on $ \succsim_i $. For every allocation $ p $, $ O'\subset O $ and $ i \in I$, we define $ p_{i}[O'] \coloneqq\sum_{o\in O'}p_{i,o}$.

We prove the lemma by contradiction. Fix a preference profile $ \succsim_I $ and two IR allocations $ p $ and $ p' $ such that $ p $ is unbidominated and $ \norm{p'}>\norm{p} $. Define $ A\coloneqq\{o\in O: \norm{p'_o}>\norm{p_o}\} $. Since $ \norm{p'}>\norm{p} $, $ A $ is nonempty. For each $ o\in A $, there exists $ i\in I$ such that $ p'_{i,o}>p_{i,o} $. Suppose $ p $ and $ p' $ are upper-equivalent. It means that for all $ i\in I $ and all $ o\in O(\ggcurly_i,\emptyset) $, $ \sum_{o'\succsim_i o}p_{i,o'}= \sum_{o'\succsim_i o}p'_{i,o'} $. In other words, for all $ i $ and all $ 1\le k<K_i $, $ p_{i}[O^k(\succsim_i)]= p'_{i}[O^k(\succsim_i)] $. 
	 We will proceed by a few steps to find a contradiction.
	
	\textbf{Step one:} We generate a directed network according to the following procedure.
	
	\textit{Step 1}: Let each $ o\in A $ point to each $ i\in I $ such that $ p'_{i,o}>p_{i,o} $. It means that we create an edge $ o\rightarrow i $. Denote by $ I_1 $ the set of agents who are pointed by an object in $ A $. We have explained that $ I_1 $ is nonempty. 
	
	Let each $ i\in I_1 $ point to each $ o'\in O $ such that (1) $ p'_{i,o'}<p_{i,o'} $ and (2) there exists $ o\in A $ such that $ o \rightarrow i $ and $ o\succsim_i o' $. It means that we create an edge $ i\rightarrow o' $. Denote by $ O_1 $ the set of objects pointed by agents in $ I_1 $. It may happen that $ O_1 \cap A \neq \emptyset $.
	
	 \textit{Step 2}: Let each $ o\in O_1 $ point to each $ i\in I $, if $ o  $ has not pointed to $ i $, such that $ p'_{i,o}>p_{i,o} $. Denote the set of such $ i $ by $ I_2 $. It may happen that $ I_2 \cap I_1 \neq \emptyset $.	
	Let each $ i\in I_2 $ point to each $ o'\in O $, if $ i $ has not pointed to $ o' $, such that (1) $ p'_{i,o'}<p_{i,o'} $ and (2) there exists $ o\in A\cup O_1 $ such that $ o \rightarrow i $ and $ o\succsim_i o' $. Denote the set of such $ o' $ by $ O_2 $.
	
	$ \ldots $
	
	 \textit{Step $ k \ge 3$}: Let each $ o\in O_{k-1} $ point to each $ i\in I $, if $ o  $ has not pointed to $ i $, such that $ p'_{i,o}>p_{i,o} $. Denote the set of such $ i $ by $ I_k $. 	
	Let each $ i\in I_k $ point to each $ o'\in O $, if $ i $ has not pointed to $ o' $, such that (1) $ p'_{i,o'}<p_{i,o'} $ and (2) there exists $ o\in A \cup [\cup_{\ell=1}^{k-1} O_{\ell}] $ such that $ o \rightarrow i $ and $ o\succsim_i o' $. Denote the set of such $ o' $ by $ O_k $.

	\textit{Stop}: The procedure stops when there are no new edges created in some step.

	Since there are finite agents and finite objects, the procedure must stop in finite steps. In the generated directed network, denote the set of agents by $ I' $ and the set of objects by $ O' $. So $ A\subseteq O' $. From the procedure we know that for every $ i\in I' $, there exists an object $ o\in A $ and a directed path from $ o $ to $ i $, written as
	\[
	o \rightarrow i_1 \rightarrow o_1 \rightarrow i_2 \rightarrow o_2 \rightarrow \cdots \rightarrow i_{n} \rightarrow o_n \rightarrow i,
	\]such that, for all $ \ell\in\{1,\ldots,n\} $, $ p'_{i_\ell,o_{\ell-1}}> p_{i_\ell,o_{\ell-1}}$, $ p'_{i_\ell,o_{\ell}}< p_{i_\ell,o_{\ell}} $, and $ o_{\ell-1} \succsim_{i_\ell} o_\ell $ ($ o_0=o $), and $ p'_{i,o_n}>p_{i,o_n} $. If $ n=0 $, $ o $ directly points to $ i $. 
	
	\medskip
	\textbf{Step two:}
	We prove that for all $ i\in I' $, $ \norm{p_i}=1 $. Suppose for some $ i\in I' $, $ \norm{p_i}<1 $. Consider the above path from some $ o\in A $ to $ i $. Given that $ p $ and $ p' $ are IR, for all $ \ell\in\{1,\ldots,n\} $,  both $ o_{\ell-1} $ and $ o_\ell $ are acceptable to $ i_\ell $ because $ p'_{i_\ell,o_{\ell-1}}>0 $ and $ p_{i_\ell,o_{\ell}}>0 $, and $ o_n $ is acceptable to $ i $ because $ p'_{i,o_n}>0 $. Since $ \norm{p_o}<\norm{p'_o} $, it must be that $ \norm{p_o}<q_o $. So for a sufficiently small $ \epsilon>0 $, in the above path if we increase every $ p_{i_\ell, o_{\ell-1}} $ by $ \epsilon $, decrease every $ p_{i_\ell, o_{\ell}} $ by $ \epsilon$, increase $ p_{i,o_n} $ by $ \epsilon $, and do not change the other probabilities, then we obtain a new allocation that bidominates $ p $, contradicting the assumption that $ p $ is unbidominated.

	\medskip
	\textbf{Step three:} We prove that $ O\backslash O'\neq \emptyset $. The fact that $ \norm{p_i}=1 $ for all $ i\in I' $ implies that $ \sum_{i\in I'} \norm{p_i}\ge \sum_{i\in I'} \norm{p'_i}$. 
	Because $ \norm{p}<\norm{p'} $, we must have $ I\backslash I'\neq \emptyset $ and $ \sum_{i\in I\backslash I'} \norm{p_i}< \sum_{i\in I\backslash I'} \norm{p'_i}$.
	 For every $ i\in I\backslash I' $ and every $ o\in O' $, because $ o $ does not point to $ i $, $ p_{i,o}\ge p'_{i,o} $. If $ O'=O $,  we should have $ \sum_{i\in I\backslash I'} \norm{p_i}\ge \sum_{i\in I\backslash I'} \norm{p'_i}$, which is a contradiction.

	\medskip
	\textbf{Step four:} We prove that there exists $ o\in O\backslash O' $ such that $ o\in A\subseteq O' $, which will be a contradiction. For every $ i\in I' $, let $ o^i $ be one of the best objects in $ O' $ such that $ o^i\rightarrow i $. That is, there does not exist $ o'\in O' $ such that $ o'\succ_i o^i $ and $ o'\rightarrow i $. For every $ o\in O\backslash O' $, if $ o $ is not acceptable to $ i $, then $ p_{i,o}=p'_{i,o}=0 $, whereas if $ o $ is acceptable to $ i $, then either $ o^i\succsim_i o $ or $ o\succ_i o^i $.  If $ o^i\succsim_i o $, because $ i $ does not point to $ o $, it must be that $ p_{i,o}\le p'_{i,o} $. So,
	\begin{equation}\label{equation1}
	\sum_{o\in O\backslash O':o^i\succsim_i o }p_{i,o}\le \sum_{o\in O\backslash O':o^i\succsim_i o}p'_{i,o}.
	\end{equation}
	If $ o \succ_io^i$, suppose $ o\in O^k(\succsim_i) $ for some $ k $. It must be that $ k<K_i $. For all $ o'\in O'\cap O^k(\succsim_i) $, since $ o' $ does not point to $ i $ (because $ o'\succ_i o^i $), we must have $ p_{i,o'}\ge p'_{i,o'} $. So
	\[
	\sum_{o'\in O^k(\succsim_i)\cap O'}p_{i,o'}\ge \sum_{o'\in O^k(\succsim_i)\cap O'}p'_{i,o'}.
	\]
	
	Recall that we have assumed that $ p_{i}[O^k(\succsim_i)]= p'_{i}[O^k(\succsim_i)] $, which is equivalent to
	\begin{align*}
	\sum_{o'\in O^k(\succsim_i)\cap O'}p_{i,o'}+\sum_{o'\in O^k(\succsim_i)\cap (O\backslash O')}p_{i,o'}&=\sum_{o'\in O^k(\succsim_i)\cap O'}p'_{i,o'}+\sum_{o'\in O^k(\succsim_i)\cap (O\backslash O')}p'_{i,o'}.
	\end{align*}
	\medskip
	
	Since $ \sum_{o'\in O^k(\succsim_i)\cap O'}p_{i,o'}\ge \sum_{o'\in O^k(\succsim_i)\cap O'}p'_{i,o'} $, we obtain
	
	\[
	\sum_{o'\in O^k(\succsim_i)\cap (O\backslash O')}p_{i,o'}\le \sum_{o'\in O^k(\succsim_i)\cap (O\backslash O')}p'_{i,o'}.
	\]
	
	Summarizing over all $ o'\in  O^k(\succsim_i)\cap (O\backslash O') $ for all $ k $ such that $ o'\succ_i o^i $, we obtain
	\begin{equation}\label{equation2}
	\sum_{o'\in O\backslash O':o'\succ_i o^i}p_{i,o'}\le \sum_{o'\in O\backslash O':o'\succ_i o^i}p'_{i,o'}.
	\end{equation}
	
	Summarizing (\ref{equation1}) and (\ref{equation2}), we obtain
	
	\[
	\sum_{o'\in O\backslash O'}p_{i,o'}\le \sum_{o'\in O\backslash O'}p'_{i,o'}.
	\]
	
	Summarizing over all $ i\in I' $, we obtain
	\begin{equation}\label{equation3}
	\sum_{i\in I'}\sum_{o'\in O\backslash O'}p_{i,o'}\le \sum_{i\in I'}\sum_{o'\in O\backslash O'}p'_{i,o'}.
	\end{equation}

	We have proved that $ \sum_{i\in I\backslash I'} \norm{p_i}< \sum_{i\in I\backslash I'} \norm{p'_i}$, which is equivalent to
	\[
	\sum_{i\in I\backslash I'}\sum_{o\in O}p_{i,o}<\sum_{i\in I\backslash I'}\sum_{o\in O}p'_{i,o},
	\]
	or,
	\[
	\sum_{i\in I\backslash I'}\sum_{o\in O'}p_{i,o}+\sum_{i\in I\backslash I'}\sum_{o\in O\backslash O'}p_{i,o}<\sum_{i\in I\backslash I'}\sum_{o\in O'}p'_{i,o}+\sum_{i\in I\backslash I'}\sum_{o\in O\backslash O'}p'_{i,o}.
	\]
	
	For every $ i\in I\backslash I' $ and every $ o\in O' $, because $ o $ does not point to $ i $, $ p_{i,o}\ge p'_{i,o} $. So
	\[
	\sum_{i\in I\backslash I'}\sum_{o\in O'}p_{i,o}\ge \sum_{i\in I\backslash I'}\sum_{o\in O'}p'_{i,o}.
	\] 
	
	So we must have 
	\begin{equation}\label{equation4}
	\sum_{i\in I\backslash I'}\sum_{o\in O\backslash O'}p_{i,o}<\sum_{i\in I\backslash I'}\sum_{o\in O\backslash O'}p'_{i,o}.
	\end{equation}
	
	Summarizing (\ref{equation3}) and (\ref{equation4}), we obtain
	\[
	\sum_{i\in I}\sum_{o\in O\backslash O'}p_{i,o}<\sum_{i\in I}\sum_{o\in O\backslash O'}p'_{i,o},
	\]
	or,
	\[
	\sum_{o\in O\backslash O'} \norm{p_o}< \sum_{o\in O\backslash O'} \norm{p'_o}.
	\]
	
	So there exists $ o\in O\backslash O' $ such that $ \norm{p_o}< \norm{p'_o} $. But this means that $ o\in A\subseteq O' $, which is a contradiction.
	
\section{Examples for Section \ref{section:size:improve}}\label{appendix:size:improve}

In Example \ref{example:PS_is_improved}, we construct a mechanism that strictly size dominates PS. The mechanism is not truncation-invariant, but it is ordinally efficient and envy-free.

\begin{example}[PS is strictly size dominated by an ordinally efficient and envy-free mechanism]\label{example:PS_is_improved}
	There are four agents $ 1,2,3,4 $ and four objects $ a,b,c,d $. Consider a mechanism $ \psi $ that is different from PS in the preference profile $ \succsim_I $ in the following table, and coincides with PS in the other preference profiles. In the table we use $ o^x $ to mean that the relevant agent obtains a probability share $ x $ of object $ o $. Unacceptable objects are omitted from agents' preference lists.
	
	\begin{table}[!htb]
		\centering
		\begin{subtable}{0.45\textwidth}
			\centering
			\begin{tabular}{llll}
				$ \succsim_1 $ & $ \succsim_2 $ & $ \succsim_3 $ & $ \succsim_4 $ \\\hline
				$  a^{1/2} $ & $ a^{1/2} $ & $  b^{1/2} $ & $ b^{1/2} $\\
				$ c^{1/4} $ & $c^{1/4} $ & $c^{1/4} $ & $ c^{1/4} $\\
				&  & $d^{1/4} $ & $ d^{1/4} $\\
				& & $ a^0 $ & $  a^0 $
			\end{tabular}
			\subcaption{$ PS(\succsim_I) $}\label{table1}
		\end{subtable}
		\begin{subtable}{0.45\textwidth}
			\centering
			\begin{tabular}{llll}
				$ \succsim_1 $ & $ \succsim_2 $ & $ \succsim_3 $ & $ \succsim_4 $ \\\hline
				$ a^{1/2} $ & $ a^{1/2} $ & $ b^{1/2} $ & $ b^{1/2} $\\
				$ c^{1/2} $ & $ c^{1/2} $ & $c^0 $ & $ c^0 $\\
				&  & $ d^{1/2} $ & $ d^{1/2} $\\
				& & $ a^0 $ & $ a^0 $
			\end{tabular}
			\subcaption{$ \psi(\succsim_I) $}\label{table2}
		\end{subtable}
	\end{table}
	
	It is easy to see that $ \psi(\succsim_I) $ is ordinally efficient and envy-free. So $ \psi $ is an ordinally efficient and envy-free mechanism. From $PS(\succsim_I)$ to $ \psi(\succsim_I)$, $ 1 $ and $ 2 $ obtain more amounts of objects, and $ 3 $ and $ 4 $ obtain the same amounts of objects as before. So $ \psi$ strictly size dominates PS. But $ \psi $ is not truncation-invariant. If $ 3 $ truncates preferences at $ c $ in $ \succsim_I $, he will obtain the lottery $ (1/2 b,1/4 c )$, instead of the lottery $ (1/2 b )$ that is required by truncation-invariance.
\end{example}

In Example \ref{example:improved by PS}, we construct an IR and unbidominated mechanism that is not truncation-invariant, but is strictly size dominated by PS.

\begin{example}[An IR and unbidominated mechanism is strictly size dominated by PS]\label{example:improved by PS}
	There are four agents $ 1,2,3,4 $ and four objects $ a,b,c,d $. Consider a mechanism $ \psi $ that is different from PS in the preference profile $ \succsim_I $ in the following table, and coincides with PS in the other preference profiles. In the table we use $ o^x $ to mean that the relevant agent obtains a probability share $ x $ of object $ o $. Unacceptable objects are omitted from agents' preference lists.
	
	\begin{table}[!htb]
		\centering
		\begin{subtable}{0.45\textwidth}
			\centering
			\begin{tabular}{llll}
				$ \succsim_1 $ & $ \succsim_2 $ & $ \succsim_3 $ & $ \succsim_4 $ \\\hline
				$  a^{1/2} $ & $ a^{1/2} $ & $  b^{1/2} $ & $ b^{1/2} $\\
				$ c^{1/4} $ & $c^{1/4} $ & $c^{1/4} $ & $ c^{1/4} $\\
				&  & $d^{1/4} $ & $ d^{1/4} $\\
				& & $ a^0 $ & $  a^0 $
			\end{tabular}
			\subcaption{$ PS(\succsim_I) $}\label{table3}
		\end{subtable}
		\begin{subtable}{0.45\textwidth}
			\centering
			\begin{tabular}{llll}
				$ \succsim_1 $ & $ \succsim_2 $ & $ \succsim_3 $ & $ \succsim_4 $ \\\hline
				$ a^{1/2} $ & $ a^{1/2} $ & $ b^{1/2} $ & $ b^{1/2} $\\
				$ c^{0} $ & $ c^{0} $ & $c^{1/2} $ & $ c^{1/2} $\\
				&  & $ d^{0} $ & $ d^{0} $\\
				& & $ a^0 $ & $ a^0 $
			\end{tabular}
			\subcaption{$ \psi(\succsim_I) $}\label{table4}
		\end{subtable}
	\end{table}
	
	Because $ \psi(\succsim_I) $ is non-wasteful, $ \psi $ is unbidominated. From $ \psi(\succsim_I)$ to $PS(\succsim_I)$, $ 1 $ and $ 2 $ obtain more amounts of objects, and $ 3 $ and $ 4 $ obtain the same amounts of objects as before. So $ \psi $ is strictly size dominated by PS. But $ \psi $ is not truncation-invariant. If $ 3 $ truncates preferences at $ c $ in $ \succsim_I $, he will obtain the lottery $ (1/2 b,1/4 c )$, instead of the lottery $ (1/2b,1/2c )$ that is required by truncation-invariance.
\end{example}

\bibliographystyle{ecta}
	
	\begin{center}
		\begingroup
		\setstretch{1.0}
		\bibliography{reference}
		\endgroup
	\end{center}

\end{document}